\documentclass{article}
\usepackage{graphicx} 
\usepackage{amsmath, amssymb, amsthm}
\usepackage{booktabs}
\usepackage{tikz}
\usepackage{xspace}
\usepackage[utf8]{inputenc}
\usepackage{cite} 
\usepackage{makecell}
\usepackage{xcolor}
\usepackage{siunitx}
\usepackage{hyperref}
\usepackage{algorithm2e}
\usepackage[font=small,labelfont=bf,margin=0.5cm,tableposition=top]{caption}
\usepackage{hhline}
\usepackage{placeins}
\usepackage{subfig}
\newcommand{\Prob}{\text{Prob}}
\newcommand{\G}{\mathcal{G}}
\newcommand{\I}{\mathcal{I}}
\newcommand{\mP}{\mathcal{P}}
\newcommand{\SO}{\mathcal{SO}}

\newcommand{\C}{\mathcal{C}}
\newcommand{\ST}{\text{ST}}
\SetKwFunction{CompleteCut}{Complete Cut}
\SetKwFunction{SimCut}{Simultaneous Cut}
\SetKwFunction{Bonsai}{Bonsai}
\SetKwFunction{MaxTrees}{MaxTrees}
\SetKwFunction{MaxFails}{MaxFails}

\DeclareMathOperator{\pop}{pop}

\theoremstyle{plain}

\newtheorem{lemma}{Lemma}
\newtheorem{proposition}{Proposition}

\theoremstyle{definition}

\newtheorem{example}{Example}

\title{Bonsai: A class of effective methods for independent sampling of graph partitions}

\author{
Jeanne Clelland\thanks{University of Colorado Bolder, Jeanne.Clelland@colorado.edu} \\
\and
Kristopher Tapp\thanks{Saint Joseph's University, ktapp@sju.edu}\\
}

\date{}

\begin{document}
\maketitle

\begin{abstract}
    We develop effective methods for constructing an ensemble of district plans via  \emph{independent} sampling from a reasonable probability distribution on the space of graph partitions.  We compare the performance of our algorithms to that of standard Markov Chain based algorithms in the context of grid graphs and state congressional and legislative maps.  For the case of perfect population balance between districts, we provide an explicit description of the distribution from which our method samples.
\end{abstract}

\section{Introduction}
Recent court cases in several states have argued about whether enacted voting maps are biased by comparing them to ensembles of thousands or millions of random maps.  The algorithms from which these ensembles are built are mostly based on spanning trees.  The ReCom (\cite{DeFord_Duchin_Solomon_2020, Cannon_Duchin_Randall_Rule_2022}) and Forest ReCom (\cite{Autry_Carter_Herschlag_Hunter_Mattingly_2023}) algorithms use Markov chains --  from a starting seed map, they repeatedly make changes to boundaries between pairs of districts.  These Markov chain methods would be slow to construct an ensemble of {\em independently} sampled plans, since this would require---assuming that the chain is ergodic, which is unknown in general---running the chain as long as the mixing time (which is also unknown in general) for each individual sampled plan.  The algorithm introduced in~\cite{McCartan_Imai_2023} uses Sequential Monte Carlo (SMC) methods rather than Markov chains, but is also ineffective at independent sampling of district plans.

The ability to effectively sample plans independently would have several advantages over existing methods, including the following.
\begin{enumerate}
    \item Independent sampling bypasses concerns about the ergodicity and mixing times of chains.  In contrast, there are few rigorous theoretical results regarding the ergodicity or mixing time of Markov chain-based methods such as ReCom, and there are examples for which ReCom is known to be slow mixing~\cite{Charikar_Liu_Liu_Vuong_2023} or \emph{not} ergodic~\cite{Tucker-Foltz_2024}, \cite{Akitaya_Korman_Korten_Souvaine_Tóth_2020},\cite{Cannon_2023}.  
    \item An algorithm that samples independently can take massive advantage of parallelization.
    \item When sampling independently, a smaller ensemble is sufficient because the effective sample size equals the actual sample size.  In contrast, autocorrelation in ReCom and repetition in SMC ensembles (caused by variable weights, ancestor extinction and other factors) can cause the effective sample size to be far smaller than the actual sample size, thereby requiring a much larger sample to achieve the desired level of statistical accuracy. 
\end{enumerate}
For a benchmark related to claim (3), note that the authors of~\cite{Clelland_Colgate_DeFord_Malmskog_Sancier-Barbosa_2021} chose their ReCom chain lengths long enough so that, for two independent chains, the expected value of the Kolmogorov-Smirnov distance between the corresponding two distributions (with respect to certain partisan metrics of interest) empirically seemed to be less than $0.01$.  They observe that for truly independent sampling, which our ``one plan at a time'' methods achieve, this stringent condition requires an ensemble size of at least $15,094$.  However, because of autocorrelation, chain lengths much longer than $15,094$ steps are required to make the \emph{effective} sample size this large.  In particular, for \emph{Reversible} ReCom chains, it is observed in~\cite{Tapp_Proebsting_Ramsay_2025} that chain lengths in the billions are not long enough in certain states.

We are pleased to thank the referees for valuable suggestions that helped improve this work.
\section{Set-up}
The typical starting point for redistricting models is a graph $\G$ whose nodes represent precincts (or other atomic units) labeled with population, and whose edges represent adjacency.  We will assume that the population of each node is positive.  A \emph{district plan} with $k$ districts is a partition of the nodes of $\G$ into $k$ sets of approximately equal population such that each of the induced subgraphs $\{D_1, ..., D_k\}$ (called \emph{districts}) is connected.  The ``approximately equal population'' requirement means that each district must be within some tolerance $\epsilon$ of the \emph{ideal population} of a district, which is defined as $\mathcal{I} = \frac{\pop(\G)}{k}$.  That is,
\begin{equation} (1-\epsilon) \I \leq \pop(D_i) \leq (1+\epsilon) \I .\label{E:popdev}\end{equation}

A well-studied distribution on the space of district plans is the \emph{spanning tree distribution}, which assigns to each plan $\mP=\{D_1,...,D_k\}$ a probability proportional to the product of the number of spanning trees on the districts.  That is,
$$\text{Prob}(\mP) \sim \prod_{i=1}^k \ST(D_i),$$
where $\ST(D_i)$ denotes the number of spanning trees on $D_i$.  In particular, ReCom empirically seems to sample from a distribution close to this, and its reversible variant introduced in~\cite{Cannon_Duchin_Randall_Rule_2022} is designed to sample exactly from the spanning tree distribution.

\section{Samplers with perfect balance ($\epsilon=0$)}
In this section, we introduce two algorithms that sample plans with perfect balance ($\epsilon=0$), followed by an equivalent implementation of the second algorithm that easily generalizes to imperfect balance.  This generalization will be introduced and named \Bonsai in the next section.
\subsection{\protect\CompleteCut}
We first introduce our simplest algorithm for independent sampling, which we call \CompleteCut.

In the case of perfect balance $(\epsilon=0)$, we call an edge $e$ of a spanning tree $T$ of $\G$ a \emph{valid cut edge} if its removal partitions $T$ into two subtrees, each of whose populations is an integer multiple of the ideal population $\I$.

A spanning tree $T$ of $\G$ is called \emph{completely cuttable} if there exists a set of $k-1$ edges whose removal partitions $T$ into $k$ subtrees (districts), each with population $\I$.  

\begin{lemma}
    $T$ is completely cuttable if and only if it has exactly $k-1$ valid cut edges.
\end{lemma}

\begin{proof}
First assume $T$ is completely cuttable, which means there are $k-1$ edges whose removal yields districts of population $\I$.  Each of these edges is a valid cut edge, and there are no other valid cut edges because removing an edge inside one of these districts $D$ would bipartition the graph into pieces each comprised of whole districts plus \emph{part} of $D$. 

Conversely, assume that $T$ has exactly $k-1$ valid cut edges.  Their removal partitions $\G$ into $k$ subgraphs $D_1,...,D_k$, and we claim that each $D_i$ has population $\I$.  To see this, consider the quotient graph in which each $D_i$ is contracted to a point.  The $k-1$ valid cut edges form a tree of this quotient graph.  Since a tree has at least one leaf, we find one district with population $\I$.  Removing it and repeating the argument gives that each $D_i$ has population $\I$.
\end{proof}

A simple way to sample a district plan is to draw trees until one of them is completely cuttable.

\RestyleAlgo{ruled}
\begin{algorithm}[H]
\caption{\protect\CompleteCut}
\label{alg:perfect-balance}
\SetKwInOut{Input}{Input}
\SetKwInOut{Output}{Output}

\Input{$\G, k$}
\Output{A connected partition of $\G$ into $k$ districts with exact balance ($\epsilon=0)$}

    \While{true}{
        $T\gets$ a uniform spanning tree of $\G$\;
        Mark every valid cut edge of $T$\;
        \If{there are $k-1$ marked edges (that is, if $T$ is completely cuttable)}{
             Remove the marked edges from $T$\;
             \Return the resulting plan\; 
        }
    }
\end{algorithm}

In the case where $\G$ is the $m$-by-$n$ rectangular grid ($m \geq n$) and $k|m$, it was shown in~\cite{Cannon_Pegden_Tucker-Foltz_2024} that a polynomial fraction $\frac{1}{p(m,n)}$ of the spanning trees of $\G$ are completely cuttable.  According to~\cite{Charikar_Liu_Liu_Vuong_2023} this implies that it is possible to sample from the spanning tree distribution on $\G$ in polynomial time. 

More trivially, this result from ~\cite{Cannon_Pegden_Tucker-Foltz_2024} implies that it is possible in polynomial time to sample from the distribution from which \CompleteCut samples.  This distribution is related to the spanning tree distribution and can be described explicitly as follows.

\begin{proposition}\label{P:completecut}    
The probability that \CompleteCut selects the partition $\mP = \{D_1,...,D_k\}$ is
$$ \text{Prob}(\mP) \sim \ST(\G/\mP) \cdot \prod_{i=1}^k \ST(D_i),$$
where $\ST(\G/\mP)$ denotes the number of spanning trees on the quotient multi-graph $\G/\mP$, 
which is defined by contracting $D_i$ to a single vertex for each $i\in[k]$.
\end{proposition} 
\begin{proof}
    A completely cuttable spanning tree that induces $\mP$ is formed by choosing a spanning tree for each district, plus choosing the connecting edges.  Choosing the connecting edges is equivalent to selecting a spanning tree of $\G/\mP$.
\end{proof}

Despite the polynomial-time results, \CompleteCut is too slow in practice to use effectively on large graphs.  We performed computational experiments with grid graphs of various sizes and numbers of districts, in which we used Wilson's algorithm to draw 1,000,000 uniform spanning trees and counted the number of valid cut edges in each tree.  Results are summarized in Table \ref{table:cuttability}.

\begin{table}[bht!]
    \centering
    \renewcommand{\arraystretch}{1.3}
    \caption{For a variety of grid graphs and district partition sizes, we sampled 1,000,000 uniform spanning trees.  Here we show (1) the percentage of trees that were completely cuttable; (2) the maximum number of valid cut edges seen in {\em any} sampled tree; (3) the number of trees which achieved the maximum number of valid cut edges.}
    \resizebox{\textwidth}{!}
    {    
\begin{tabular}{|c|c|c|c|c|}
\hline
Size of  & Number of  & Pct completely  & Max number of & Trees with max \\[-0.05in]
grid graph & districts & cuttable & valid cut edges & valid cut edges \\
\hline
$7 \times 7$ & $7$ & $1.58\%$ & $6$ & $15,749$\\
\hline
$50 \times 50$ & $5$ & $0.013\%$ & $4$ & $134$ \\
\hline
$50 \times 50$ & $10$ & $0\%$ & $8$ & $2$ \\
\hline
$50 \times 50$ & $25$ & $0\%$ & $16$ & $1$ \\
\hline
$50 \times 50$ & $50$ & $0\%$ & $27$ & $2$ \\
\hline
    \end{tabular}
}
    \label{table:cuttability}
\end{table}
    
As Table \ref{table:cuttability} shows, the percentage of completely cuttable trees drops precipitously as the size of the graph increases, even if the number of districts remains small. And as the number of districts increases, the maximum number of valid cut edges in {\em any} sampled tree decreases rapidly relative to the number required for complete cuttability.

For plans where imperfect population balance is permitted, we might initially hope for better success with this strategy.  In particular, it is shown in ~\cite{Cannon_Pegden_Tucker-Foltz_2024} that for an $m \times n$ grid graph ($m \geq n$) and $k | m$, the fraction of uniform spanning trees that can be partitioned into $k$ districts $D_i$ whose populations satisfy Equation~\ref{E:popdev}
with $\frac{400 k}{n} \leq \epsilon \leq \frac{2}{3}$
is bounded below by a {\em constant}. (Note that this condition requires that $n \geq\frac{400 k}{\epsilon}$, which limits the applicability of this result in practice.)  Empirical results are obtained on grid graphs of sizes $50 \times 50$ and $100 \times 100$ for $k=2$; in both cases it appears that slightly less than $1\%$ of uniformly sampled trees can be bipartitioned with $\epsilon \leq 0.05$, a typical value in redistricting applications.  But if $k$ is much larger than $2$, then the fraction of trees with $k-1$ valid cut edges is likely to remain impracticably small.  

Moreover, when $\epsilon>0$, it is a subtle problem to efficiently decide whether there exist $k-1$ simultaneously valid cut edges; an algorithm for this is found in~\cite{BUD}.


\subsection{\protect\SimCut}
We next introduce the algorithm \SimCut that improves on \CompleteCut because, instead of waiting for a completely cuttable tree, it makes all possible cuts for the given tree, and then repeats the process on the pieces that require further cuts to become single districts.

\RestyleAlgo{ruled}
\begin{algorithm}[H]
\caption{\protect\SimCut}
\label{alg:SimCut}

\SetKwInOut{Input}{Input}
\SetKwInOut{Output}{Output}

\Input{$\G, k$}
\Output{A connected partition of $\G$ into $k$ districts with exact population balance ($\epsilon = 0$)}

\SetKwFunction{GenPlan}{GeneratePlan}
\SetKwProg{Fn}{Function}{:}{}

\BlankLine
$\mathcal{I} \gets \frac{\pop(\G)}{k}$

\Return $\GenPlan(\G)$\;

\BlankLine
\Fn{\GenPlan{$H$}}{
    \If{$\pop(H)=\mathcal{I}$}{
        \Return $\{H\}$\;
    }

    \While{true}{
        $T\gets$ a uniform spanning tree of $H$\;

        Mark every cut edge of $T$ that is valid (with respect to $\mathcal{I})$;

        \If{no edges are marked}{
            continue\;
        }

        Remove all marked edges from $T$\;
        Let $\{H_1, \dots, H_l\}$ be the resulting subgraphs\;
        \Return $\GenPlan(H_1)\cup\cdots\cup \GenPlan(H_l)$\;
    }
}
\end{algorithm}

\subsection{The distribution from which \protect\SimCut samples}
In this subsection, we derive an explicit formula for the distribution from which the \SimCut algorithm samples.  That is, we let $\mP = \{D_1, ..., D_k\}$ be a district plan, and we compute the probability that \SimCut outputs $\mP$.

For this, we will need to sum over the different \emph{splitting orders} that result in $\mP$.  For example, when $k=8$, figure~\ref{F:partition_sequence} shows one potential splitting order, $\SO$,  that splits $\G$ into the districts $\{D_1,...,D_8\}$ in the following order:
\begin{align}\label{E:splitting_order}
\{1,...,8\} 
& \rightarrow \{{\color{red}\{}1,...,7{\color{red}\}},{\color{red}\{}8{\color{red}\}}\} \\
& \rightarrow \{{\color{red}\{} {\color{orange}\{}1,2{\color{orange}\}},{\color{orange}\{}3,4,7{\color{orange}\}},{\color{orange}\{}5,6{\color{orange}\}}{\color{red}\}}, {\color{red}\{}8{\color{red}\}}\} \notag\\ 
& \rightarrow  \{{\color{red}\{} {\color{orange}\{}1,2{\color{orange}\}},{\color{orange}\{}
{\color{green}\{}3,4{\color{green}\}},{\color{green}\{}7{\color{green}\}}{\color{orange}\}},{\color{orange}\{}\{5\},\{6\}{\color{orange}\}}{\color{red}\}},
{\color{red}\{}8{\color{red}\}}\}\notag\\ 
& \rightarrow \{{\color{red}\{}{\color{orange}\{}\{1\},\{2\}{\color{orange}\}},{\color{orange}\{}
{\color{green}\{}\{3\},\{4\}{\color{green}\}},{\color{green}\{}7{\color{green}\}}{\color{orange}\}},{\color{orange}\{}\{5\},\{6\}{\color{orange}\}}{\color{red}\}}, {\color{red}\{}8{\color{red}\}}\}. \notag
\end{align}

\begin{figure}[bht!]\centering
\includegraphics[width=1in]{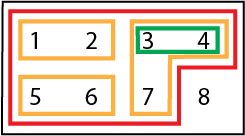}
\caption{A splitting order $\SO$ of $\{1,...,8\}$}\label{F:partition_sequence}
\end{figure}

More precisely, a \emph{splitting order} $\SO$ of $\mP$ is a sequence of successive partitions (refinements) of $\{1,...,k\}$ down to singletons.  Each step of a splitting order (each arrow in Equation~\ref{E:splitting_order}) involves partitioning (refining) some of the existing sets.  An individual such refinement, $A_1\cup\cdots\cup A_l \rightarrow \{A_1,...,A_l\}=A$, will be called a  
\emph{split} of $\SO$, and the individual pieces $A_1,...,A_l$ of all of the splits of $\SO$ will be called the \emph{collections} of $\SO$.  In the example, ${\color{red}\{}1,...,7{\color{red}\}} \rightarrow {\color{red}\{}{\color{orange}\{}1,2{\color{orange}\}},{\color{orange}\{}3,4,7{\color{orange}\}},{\color{orange}\{}5,6{\color{orange}\}}{\color{red}\}} = {\color{red}\{}A_1,A_2,A_3{\color{red}\}} = A$ is one split of $\SO$, while the collections of $\SO$ are: 
$$\{1,...,7\}, \{1,2\}, \{3,4,7\}, \{5,6\}, \{3,4\}, \{1\}, \{2\},\{3\},\{4\},\{5\},\{6\},\{7\},\{8\}. $$
For each collection $S=\{i_1,...,i_z\}$ of $\SO$, the corresponding graph $D_{i_1}\cup\cdots\cup D_{i_z}$ must be connected.

Since we intend to sum over the splitting orders of a plan $\mP$, we must clarify that two splitting orders are considered the same if they only differ regarding the step at which \emph{independent} refinements occur.  So in the previous example above, $\SO$ would be unchanged if the split 
${\color{orange}\{}\{1,2\}{\color{orange}\}} \rightarrow {\color{orange}\{}\{1\},\{2\}{\color{orange}\}}$ had occurred at step 3 instead of step 4, and similarly for the split ${\color{orange}\{}\{5,6\}{\color{orange}\}} \rightarrow {\color{orange}\{}\{5\},\{6\}{\color{orange}\}}$.

If $A = \{A_1,...,A_l\}$ is a split of $\SO$, then we denote by 
$$\C(A) =  \C(A_1|  A_2| \cdots |  A_l)$$ 
the number of spanning trees on the quotient multi-graph corresponding to $A$.  This means the graph whose nodes are $\{1,...,l\}$, and the number of edges between a pair $x,y$ of distinct nodes equals
$$\left|\left\{(a,b)\in E(\G)\mid a\in D_i, b\in D_j, i\in A_x, j\in A_y \right\} \right|.$$
Note that this quotient multigraph is connected by hypothesis.

For $i,j\in\{1,...,k\}$ with $i\neq j$, let $e_{ij}$ denote the number of edges of $\G$ connecting $D_i$ to $D_j$, which equals zero if the district pair is not adjacent.  It is possible to express $\C(A)$ fully in terms of these quantities.  For example, $\C(\{i\}|\{j\}) = e_{ij}$,
while $\C(\{i\}|\{j\}|\{k\}) = e_{ij}e_{ik} + e_{ij}e_{jk} + e_{ik}e_{jk}$.  In general, by the Matrix-Tree Theorem, $\C(A)$ is the determinant of any cofactor of the weighted Laplacian of this quotient multigraph, and hence can be computed in polynomial time in $|A|$.  Its fully expanded expression in the variables $e_{ij}$, however, may contain exponentially many terms.

If $S=\{i_1,...,i_z\}$ is a collection from $\SO$ of size $z>1$, we define $\Omega(S)$ as the odds that a UST, $T$, of $D_{i_1}\cup\cdots\cup D_{i_z}$ is un-splittable (which means it contains no valid cut edges).  Formally,
$$\Omega(S) = \frac{\ST_0(S)}{\ST_+(S)},$$
where $\ST_0(S)$, $\ST_+(S)$ respectively denote the number of un-splittable and splittable spanning trees of $D_{i_1}\cup\cdots\cup D_{i_z}$.  Note that $ST_+(S)>0$  because $\SO$ is a valid splitting sequence of $\{1,...,k\}$ down to singletons.  We will additionally express the number of total spanning trees as $\ST(S)=\ST_0(S)+\ST_+(S)$.

\begin{proposition}\label{P:P2} The probability that the algorithm produces $\mP$ equals 
\begin{align*}
    \text{Prob}(\mP) = \frac{\prod_{i=1}^k ST(\{i\})}{ST_+(\{1,...,k\})}\cdot 
    \sum_{\SO}\left(\prod_A \C(A)\right)\left(\prod_S\Omega(S) \right),
\end{align*}
where the sum is over all splitting orders $\SO$ of $\mP$, the first product is over all splits $A$ of $\SO$, and the second product is over all non-singleton collections $S$ of $\SO$.
\end{proposition}

\begin{proof}
Let $\SO$ be a splitting order that performs a total of $s$ splits, and let $P_i$ for $i\in [s]$ be the probability that the algorithm will make the $i$th split of $\SO$ when called on the multi-district on which this split operates.  The algorithm's probability of producing the plan $\mP$ in the order $\SO$ equals
\begin{equation}\label{E:now}\text{Prob}(\mP\text{ and } \SO) =\prod P_i.\end{equation}

If the $i$th split has the form $A_1\cup\cdots\cup A_l\mapsto\{A_1,...,A_l\}=A$, where the $A_j$s are disjoint subsets of $[k]$, then
\begin{equation}\label{E:Pi}P_i = \frac{\left(ST_0(A_1)\cdots ST_0(A_l)\right)\cdot \C(A)}{ST_+(A_1\cup\cdots\cup A_l)}.\end{equation}
This is because, when asked to split the multi-district $M$ on which the $i$th split operates, namely 
$$M=\cup \{D_j\mid j\in A_1\cup\cdots\cup A_l\},$$ the algorithm draws uniformly among splittable spanning trees of $M$ (by drawing spanning trees uniformly and rejecting the unsplittable ones).  Such a tree will split $M$ into exactly the pieces $M_1,...,M_l$ of the $i$th split, namely $M_j=\cup \{D_k\mid k\in A_j\}$ for $j\in[l]$, if it is comprised of an unsplittable spanning tree on each $M_i$ connected together via a spanning tree of the corresponding quotient multigraph.

After substituting Equation~\ref{E:Pi} into Equation~\ref{E:now}, every non-singleton $ST_0$ term in the numerator will have a corresponding $ST_+$ term in the denominator for a subsequent split, yielding

$$\text{Prob}(\mP\text{ and } \SO) =
\frac{\prod_{i=1}^k ST(\{i\})}{ST_+(\{1,...,k\})}\cdot 
    \left(\prod_A \C(A)\right)\left(\prod_S\Omega(S) \right).$$
The result follows by summing over the splitting orders.
\end{proof}

We now illustrate the proof with an example.
\begin{example}
When $k=4$, one possible splitting order, $\SO$, is
$$\{1,2,3,4\} 
\stackrel{P_1}{\rightarrow} \{ {\color{red}\{}1,2,3{\color{red}\}},{\color{red}\{}4{\color{red}\}}\} 
\stackrel{P_2}{\rightarrow} \{ {\color{red}\{}{\color{orange}\{}1,2{\color{orange}\}},{\color{orange}\{}3{\color{orange}\}}{\color{red}\}},{\color{red}\{}4{\color{red}\}}\} 
\stackrel{P_3}{\rightarrow} \{ {\color{red}\{}{\color{orange}\{}{\color{green}\{}1{\color{green}\}},{\color{green}\{}2{\color{green}\}}{\color{orange}\}},{\color{orange}\{}3{\color{orange}\}}{\color{red}\}},{\color{red}\{}4{\color{red}\}}\}. 
$$
Assume that $\mP$ is a plan for which $\SO$ is valid.  The probability that the algorithm produces $\mP$ in the order $\SO$ is: $$\text{Prob}(\mP\text{ and } \SO) =P_1P_2P_3,$$ where $P_i$ is the probability of the corresponding split, as labeled in the previous equation.  These probabilities are computed as follows.
\begin{align}
    P_1 & =\frac{\ST_0(\{1,2,3\})\cdot \ST(\{4\})\cdot\C(\{1,2,3\}| \{4\}) } {\ST_+(\{1,2,3,4\})} \\
    P_2 & = \frac{\ST_0(\{1,2\})\cdot\ST(\{3\})\cdot\C(\{1,2\}| \{3\})}{\ST_+(\{1,2,3\})} \\
    P_3 & = \frac{\ST(\{1\})\cdot \ST(\{2\})\cdot\C(\{1\}| \{2\})}{\ST_+(\{1,2\})}.
\end{align}
Therefore,
\begin{align*}
   \text{Prob}(\mP\text{ and } \SO)
    = \frac{\prod_{i=1}^4 ST(\{i\})}{ST_+(\{1,...,4\})}
    &\cdot \C(\{1,2,3\}|  \{4\})\cdot \C(\{1,2\}| \{3\})\cdot\C(\{1\}|  \{2\}) \\
    &\cdot\Omega(\{1,2,3\})\cdot \Omega(\{1,2\}) \\
    = \frac{\prod_{i=1}^4 ST(\{i\})}{ST_+(\{1,...,4\})}
    &\cdot \left(\prod_A \C(A)\right)\left(\prod_S\Omega(S) \right)
\end{align*}
The other splitting orders work similarly, so $\text{Prob}(\mP)$ can be found by summing $\text{Prob}(\mP\text{ and } \SO)$ over all of them, which yields the equation in Proposition~\ref{P:P2}.

As previously mentioned, $\prod_A \C(A)$ can be expressed in terms of the values $\{e_{ij}\}$.  For example in the split considered above, we have
\begin{align*} 
\C(\{1,2,3\}| \{4\})\cdot \C(\{1,2\}| \{3\})\cdot\C(\{1\}|  \{2\}) =
\left(e_{14}+e_{24}+e_{34}\right)\left(e_{13} + e_{23}\right)\left(e_{12}\right).
\end{align*}
\end{example}

\subsection{Importance sampling with \protect\SimCut}\label{sec:Importance}
The formula of Proposition~\ref{P:P2} contains the factor $\prod_{i=1}^k ST(\{i\})$, which is proportional to the probability of $\mP$ under the spanning tree distribution, which suggests correcting the remaining factors by importance sampling, and thereby using \SimCut to sample statistics from the spanning tree distribution.

We implement this strategy in the remainder of this section, but we first note the limitation of this approach: the \SimCut algorithm only applies in the case of perfect population balance ($\epsilon=0$), and it does not have the soon-to-be introduced backtracking feature, and hence can enter into a branch from which it cannot complete a valid plan.  Nevertheless, Section~\ref{sec:Importance_eprirics} of the appendix contains empirical evidence that the method works effectively with grid graphs.

For a splitting order $\SO$ of a plan $\mP$, define
\[
F(\SO)
=
\left(\prod_A \C(A)\right)
\left(\prod_S \Omega(S)\right),
\]
where the products are over the splits $A$ and the non-singleton collections
$S$ of $\SO$, respectively.  Proposition~\ref{P:P2} gives that the probability that \SimCut samples $\mP$ in the order $\SO$ is
\[
\Prob(\mP,\SO)
=
\frac{\prod_{i=1}^k ST(\{i\})}
     {ST_+(\{1,\ldots,k\})}
F(\SO).
\]
The denominator $ST_+(\{1,\ldots,k\})$ is independent of $\mP$.

Rather than summing $F(\SO)$ over all splitting orders of $\mP$, we
perform importance sampling on the augmented space of pairs $(\mP,\SO)$.
For each plan $\mP$, let $q(\SO\mid\mP)$ be a probability distribution
with full support on its valid splitting orders.  Consider the target
distribution on pairs given, up to normalization, by
\[
\pi(\mP,\SO)
\propto
\left(\prod_{i=1}^k ST(\{i\})\right)q(\SO\mid\mP).
\]
Since
\[
\sum_{\SO}q(\SO\mid\mP)=1,
\]
the marginal distribution of $\mP$ under $\pi$ is exactly the spanning
tree distribution.  Comparing this target distribution with the
distribution of $(\mP,\SO)$ produced by \SimCut shows that the
importance weight is, up to a common multiplicative constant,
\begin{equation}\label{eq:importance-weight}
w(\mP,\SO)
=
\frac{q(\SO\mid\mP)}{F(\SO)}
=
\frac{q(\SO\mid\mP)}
{\left(\prod_A \C(A)\right)
 \left(\prod_S\Omega(S)\right)}.
\end{equation}

We next choose $q$ in a way that substantially simplifies these weights.
Given $\mP=\{D_1,\ldots,D_k\}$, let $Q=\G/\mP$ denote its district-level
quotient multigraph.  More generally, for any connected collection
$S\subseteq\{1,\ldots,k\}$, let $Q_S$ denote the subgraph of $Q$
induced by the districts indexed by $S$.

We generate a splitting order recursively.  Beginning with
$S=\{1,\ldots,k\}$, draw a uniform spanning tree of $Q_S$, and choose
uniformly at random a nonempty subset of its $|S|-1$ edges.  Removing
these edges partitions $S$ into connected collections
$A_1,\ldots,A_l$.  We then apply the same procedure independently to
each non-singleton $A_j$, continuing until all collections are
singletons.  This procedure has full support on the valid splitting
orders of $\mP$.

Suppose that at one stage the procedure splits $S$ into
$A=\{A_1,\ldots,A_l\}$.  The probability of obtaining this particular
split is
\begin{equation}\label{eq:q-one-split}
\frac{\C(A)\prod_{j=1}^l ST(Q_{A_j})}
     {ST(Q_S)}
\frac{1}{2^{|S|-1}-1}.
\end{equation}
Indeed, $\C(A)\prod_j ST(Q_{A_j})$ counts the spanning trees of $Q_S$,
together with the unique set of $l-1$ edges whose removal produces
precisely the pieces $A_1,\ldots,A_l$.  Every spanning tree of $Q_S$
has $|S|-1$ edges and therefore has $2^{|S|-1}-1$ nonempty subsets of
edges.

Multiplying Equation~\ref{eq:q-one-split} over the splits of $\SO$, the
spanning-tree terms telescope. Since $ST(Q_{\{i\}})=1$, and ignoring a
multiplicative factor depending only on $k$, we obtain
\begin{equation}\label{eq:q-splitting-order}
q(\SO\mid\mP)
\propto
\frac{1}{ST(\G/\mP)}
\left(\prod_A \C(A)\right)
\left(\prod_S\frac{1}{2^{|S|-1}-1}\right),
\end{equation}
where the first product is over the splits of $\SO$, and the second is
over all collections that are not singletons and are not the full set
$\{1,\ldots,k\}$.

Substituting Equation~\ref{eq:q-splitting-order} into
Equation~\ref{eq:importance-weight} cancels all of the $\C(A)$ factors,
yielding
\begin{equation}\label{eq:simplified-importance-weight}
\boxed{
w(\mP,\SO)
=
\frac{1}{ST(\G/\mP)}
\prod_S
\frac{1}{(2^{|S|-1}-1)\Omega(S)}
}
\end{equation}
up to a multiplicative constant common to all samples.  Thus the
importance weights require no computation of the $\C(A)$ terms.  

The quantity $ST(\G/\mP)$ can be computed efficiently using the Matrix-Tree
Theorem.  The quantities $\Omega(S)$ are more difficult, but each can
be estimated by independently sampling uniform spanning trees of the
corresponding multi-district and estimating the odds that such a tree
is unsplittable.

Finally, suppose that
$(\mP_1,\SO_1),\ldots,(\mP_N,\SO_N)$ are independent samples produced
by \SimCut, with corresponding weights $w_1,\ldots,w_N$.  For any
statistic $f$ on district plans, such as the number of cut edges, its
expectation under the spanning tree distribution can be estimated by
the self-normalized importance-sampling estimator
\[
\frac{\sum_{j=1}^N w_j f(\mP_j)}
     {\sum_{j=1}^N w_j}.
\]
Alternatively, the weighted sample can be converted into an unweighted
ensemble approximating the spanning tree distribution by resampling
$\mP_1,\ldots,\mP_N$ with probabilities proportional to
$w_1,\ldots,w_N$.
\subsection{Re-describing and adding backtracking to \protect\SimCut}
The following is an equivalent formulation of Algorithm~\ref{alg:SimCut} that is phrased in a way that will more naturally generalize to the case of imperfect population balance in the next section.  Instead of simultaneously making all possible cuts from the tree $T$, it only removes a single balanced cut edge.  But it keeps the resulting subtrees on the two pieces to use as the first trees it tries when it further subdivides those pieces.  Therefore, just like before, it eventually makes all possible cuts from $T$ before drawing any new trees.  

\RestyleAlgo{ruled}
\begin{algorithm}[H]
\caption{Equivalent to \protect\SimCut}
\label{alg:Simcut2}

\SetKwInOut{Input}{Input}
\SetKwInOut{Output}{Output}

\Input{$\G, k$}
\Output{A connected partition of $\G$ into $k$ districts with exact population balance ($\epsilon = 0$)}

\SetKwFunction{GenPlan}{GeneratePlan}
\SetKwProg{Fn}{Function}{:}{}

\BlankLine
$\mathcal{I} \gets \frac{\pop(\G)}{k}$

$T \gets$ a uniform spanning tree of $\G$\;
\Return $\GenPlan(\G, T)$\;

\BlankLine

\Fn{\GenPlan{$H, T$}}{
    \If{$\pop(H)=\mathcal{I}$}{
        \Return $\{H\}$\;
    }

    \While{true}{
        Mark every cut edge of $T$ that is valid (with respect to $\mathcal{I})$\;

        \If{no edges are marked}{
            $T \gets$ a new uniform spanning tree of $H$\;
            continue\;
        }

        Remove the ``best'' marked edge from $T$\;
        Let $\{H_1, H_2\}$ be the resulting subgraphs and
        $\{T_1, T_2\}$ the induced spanning trees\;

        \Return
        $\GenPlan(H_1, T_1)\cup\GenPlan(H_2,T_2)$\;
    }
}
\end{algorithm}

Since we are assuming $\epsilon=0$, it does not matter how we define the ``best'' marked edge.  We could make a random selection, or use any other strategy, and Algorithm~\ref{alg:Simcut2} will still be equivalent to Algorithm~\ref{alg:SimCut}.

Before addressing the $\epsilon>0$ case, for which it will matter how we define the ``best'' marked edge, it is useful to first add a backtracking feature to Algorithm~\ref{alg:Simcut2} to reduce the likelihood that it becomes stuck.  In fact, the $\Bonsai$ algorithm that we will introduce in the next section, in the special case $\epsilon=0$, will be Algorithm~\ref{alg:Simcut2} with the addition of this backtracking feature.

The feature relies on two parameters.  First, \MaxTrees is the maximum number of spanning trees of a graph that the algorithm is willing to draw before backtracking.  Second, \MaxFails is the maximum number of failed attempts at completely partitioning a graph before the algorithm backtracks.  For example, if a $10$-district-size piece is split into a $7$-district-size piece and a $3$-district-size piece, but it then fails to find any valid cut edges for the $3$-district-size piece after \MaxTrees attempted trees, then it backtracks to the step of re-dividing the original $10$-district-size piece.  If this happens \MaxFails times (backtracking to the step of re-dividing the $10$-district-size piece because of downstream failures), then the algorithm further backtracks to the step of re-splitting the larger piece from which the $10$-district-size piece was cut.

Empirically, we found that $\MaxTrees=10$ and $\MaxFails=3$ worked reasonably well for real-world problems, although the optimal choices can depend on the problem.

Why did we choose in Algorithm~\ref{alg:Simcut2}, after making a cut, to keep the trees on each piece and use them for the initial attempt to further subdivide each piece?  This decision gives a moderate speed improvement (by avoiding drawing new trees unnecessarily).  More importantly, this is what makes Algorithm~\ref{alg:Simcut2} equivalent to Algorithm~\ref{alg:SimCut} in the case $\epsilon=0$ (when backtracking is not used), and we have an exact description of the sampling distribution for Algorithm~\ref{alg:SimCut}.  On the other hand, if the algorithm were to draw new trees on the two pieces after a cut, we do not believe this change would have an empirically noticeable effect on the distribution from which it samples in real-world problems.  

\begin{example}\label{bonsai-example-1}
Suppose that we want to divide a $6 \times 6$ grid graph into $6$ districts, each of size $6$.  We start by sampling a uniform spanning tree; suppose that we sample the tree shown in Figure \ref{fig:bonsai}(a).  This tree has 3 valid cut edges, shown in red in Figure \ref{fig:bonsai}(b).  Cutting these 3 edges partitions the graph into 4 pieces of sizes $6, 6, 12, 12$, as shown in Figure \ref{fig:bonsai}(c).  We then sample new trees on each of the pieces of size 12, as in Figure \ref{fig:bonsai}(d).  If these trees have valid cut edges as in Figure \ref{fig:bonsai}(e), we remove these edges to complete the district plan, as in Figure \ref{fig:bonsai}(f).

\begin{figure}[bht!]
\begin{center}
\caption{\protect\Bonsai for partitioning a $6 \times 6$ grid graph into $6$ equal-size districts: (a) uniformly sampled spanning tree; (b) valid cut edges on spanning tree; (c) partial partition of graph; (d) uniformly sampled spanning trees on double-district-sized-pieces; (e) valid cut edges on spanning trees; (f) completed partition into districts.}
\label{fig:bonsai}
\subfloat[]{\includegraphics[height=1.5in]{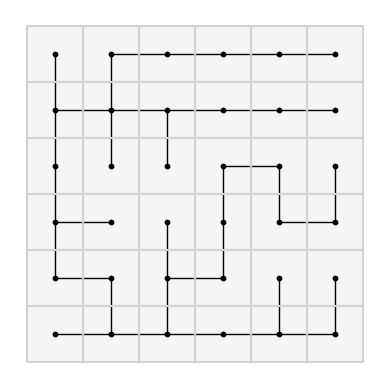} }
\ 
\subfloat[]{\includegraphics[height=1.5in]{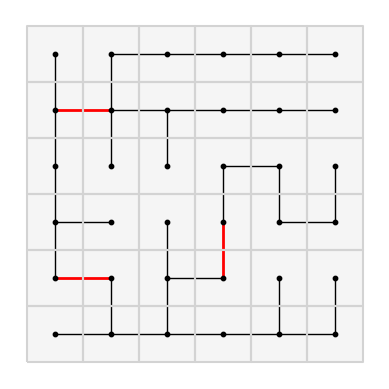}  }
\ 
\subfloat[]{\includegraphics[height=1.5in]{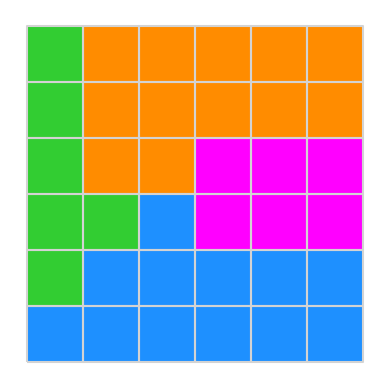}  }
\ 
\subfloat[]{\includegraphics[height=1.5in]{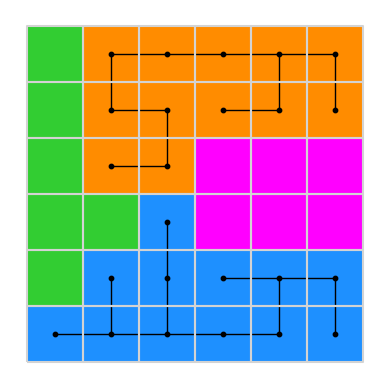}  }
\ 
\subfloat[]{\includegraphics[height=1.5in]{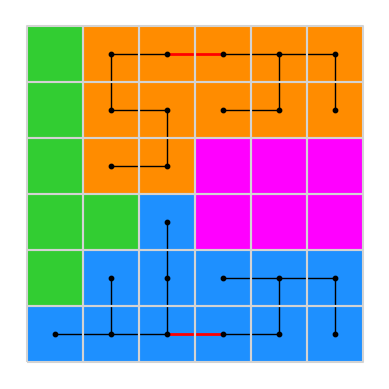} }
\ 
\subfloat[]{\includegraphics[height=1.5in]{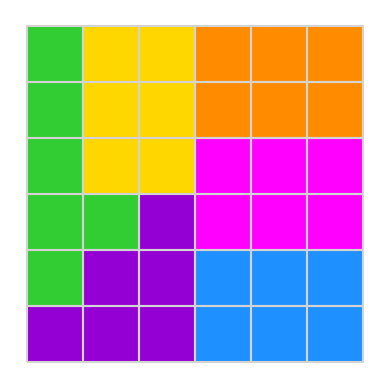} }
\end{center}
\end{figure}

\end{example}

\begin{example}\label{bonsai-example-2}

With the same scenario as in Example \ref{bonsai-example-1}, suppose that during the process of partitioning the graph into pieces, a piece is created that cannot be partitioned into equal-size pieces, as in Figure \ref{fig:bad-bonsai}(a).  Then the backtracking provisions eventually undo the previous cut that created this piece, as in Figure \ref{fig:bad-bonsai}(b); then a new tree is drawn on the merged piece and the algorithm proceeds as before, as shown in Figure \ref{fig:bad-bonsai}(c-e).

\begin{figure}[bht!]
\begin{center}
\caption{When \protect\Bonsai gets stuck: (a) a partial partition that contains an uncuttable piece; (b) uncuttable piece merged with piece from prior cut; (c) new uniformly sampled spanning tree on merged piece; (d) valid cut edges on new spanning tree; (e)  completed partition into districts.}
\label{fig:bad-bonsai}
\subfloat[]{\includegraphics[height=1.5in]{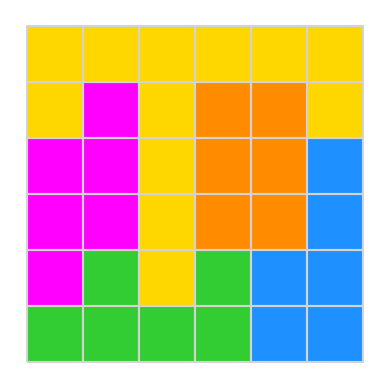} }
\ 
\subfloat[]{\includegraphics[height=1.5in]{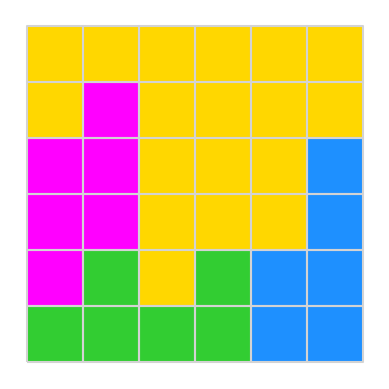}  }
\ 
\subfloat[]{\includegraphics[height=1.5in]{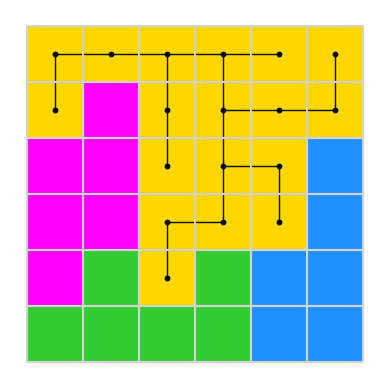}  }
\ 
\subfloat[]{\includegraphics[height=1.5in]{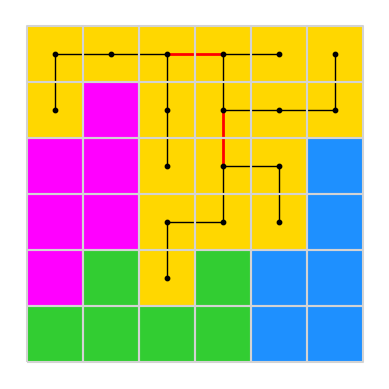}  }
\ 
\subfloat[]{\includegraphics[height=1.5in]{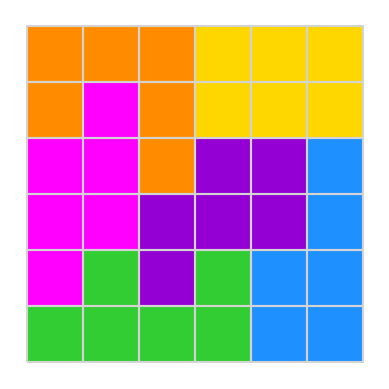} }
\end{center}
\end{figure}

\end{example}

\FloatBarrier

\section{\protect\Bonsai:\! A practical algorithm for independent sampling}

In this section we develop our algorithm for independent sampling of graph partitions, which we have named \Bonsai after the traditional Japanese art of growing and shaping trees. 

\subsection{Definition of \protect\Bonsai}
In the case $\epsilon=0$, Algorithm~\ref{alg:Simcut2} is equivalent to Algorithm~\ref{alg:SimCut} (when backtracking is not used), but it has the advantage of generalizing to the case of imperfect population balance ($\epsilon>0$) with only minor modifications to the {\tt GeneratePlan} function.  For this, we need to define the notions of ``valid cut edge'' and ``best'' in this setting, which is the goal of this section.

Let $H$ be a graph that is to be partitioned into $k$ districts. Let $T$ be a spanning tree of $H$, and suppose that the removal of an edge $e$ of $T$ partitions $H$ into two subgraphs $H_1$ and $H_2$, which are to be further partitioned into $k_1$ and $k_2$ districts, respectively, with  
$k_1+k_2=k$. 
For each $i \in \{1,2\}$, let 
\begin{equation} 
\delta_i = \frac{|\mathrm{pop}(H_i)-k_i\cdot\I|}{\mathcal{I}}.\label{define-delta}
\end{equation} 
In order for $e$ to be a valid cut edge, we need $\delta_i$ to be sufficiently small so that $H_i$ is right-sized to be further broken down into $k_i$ districts.

The loosest possible requirement is $\delta_i\leq k_i\cdot\epsilon$.  However, in real-world problems this strategy is likely to get the algorithm stuck.  If $\delta_i$ is very close to $k_i\cdot\epsilon$, then $H_i$ is barely able to be subdivided into $k_i$ districts, and it might be very unlikely or impossible for the algorithm to succeed at further breaking down $H_i$.  A much stricter requirement is $\delta_i\leq\epsilon$, which is more likely to leave enough population slack in each piece to allow it to be further broken down.  

We will interpolate between the loose and strict requirements via an arbitrary ``tolerance multiplier function'' $\phi:\mathbb{N}\rightarrow [1,\infty)$, by defining $e$ to be a \emph{valid cut edge} with respect to $\phi$ and the pair $(k_1, k_2)$ if 
$$ k_1 + k_2 = k,\,\text{ and }\,\delta_i\leq \phi(k_i)\cdot\epsilon\, \text{ for each }i\in\{1,2\}.$$
Note that $\phi(n)=n$ and $\phi(n)=1$ correspond to the above-mentioned loose and strict requirements respectively; a general tolerance multiplier must lie between these extremes:
$$ 1 \leq \phi(n) \leq n \text{ for all }n\geq 1.$$
We note that if $\epsilon$ and $\phi$ are large enough, it is possible that a cut edge $e$ may be valid with respect to more than one pair $(k_1, k_2)$.


With this definition of \emph{valid cut edge}, and once we choose a method of selecting the ``best'' valid cut edge and corresponding pair $(k_1,k_2)$, the only modification required to Algorithm~\ref{alg:Simcut2} (including the backtracking feature) in the $\epsilon>0$ setting is to keep track of the number of districts $k$ that each graph $H$ is to be partitioned into.  This is the algorithm that we use for the empirical results in the remainder of the paper.  For readability, the description below of Algorithm~\ref{alg:imperfect-balance3} suppresses the backtracking mechanism described above, which is included in our implementation.

\RestyleAlgo{ruled}
\begin{algorithm}[H]
\caption{\protect\Bonsai}
\label{alg:imperfect-balance3}

\SetKwInOut{Input}{Input}
\SetKwInOut{Output}{Output}

\Input{$\G, k, \epsilon, \phi$}
\Output{A connected partition of $\G$ into $k$ districts $H_i$ whose populations satisfy Equation~\ref{E:popdev}.
} 

\SetKwFunction{GenPlan}{GeneratePlan}
\SetKwProg{Fn}{Function}{:}{}

\BlankLine
$\mathcal{I} \gets \frac{\pop(\G)}{k}$

$T \gets$ a uniform spanning tree of $\G$\;
\Return $\GenPlan(\G, T, k, \epsilon, \phi)$\;

\BlankLine

\Fn{\GenPlan{$H, T, k, \epsilon, \phi$}}{
    \If{$k=1$}{
        \Return $\{H\}$\;
    }

    \While{true}{
        Record every triple $(e,k_1,k_2)$ for which $e$ is a cut edge of $T$ that is valid with respect to $\mathcal{I},k_1,k_2, \phi$\;

        \If{no triples are recorded}{
            $T \gets$ a new uniform spanning tree of $H$\;
            continue\;
        }

        Choose the ``best'' triple $(e,k_1,k_2)$, and remove $e$ from $T$\;
        Let $\{H_1, H_2\}$ be the resulting subgraphs and
        $\{T_1, T_2\}$ the induced spanning trees\;

        \Return
        $\GenPlan(H_1, T_1, k_1, \epsilon, \phi)\cup\GenPlan(H_2,T_2, k_2, \epsilon, \phi)$\;
    }
}
\end{algorithm}

For the empirical results in the next section, we use $\phi=1$ (the constant function), and our method of choosing the ``best'' triple $(e,k_1,k_2)$ is to select the one that is most balanced in the following sense. We first identify the triples that minimize $|k_1-k_2|$, and then among those, choose the one that minimizes $\max\{\delta_1, \delta_2\}$, breaking ties randomly.

When $\phi=1$ and $\epsilon<.25$, which will be the case in all of our experiments, it is worth noting that Algorithm~\ref{alg:imperfect-balance3} can be simplified because for each cut edge $e$, there is only a single valid choice of $k_i$ for each $i\in\{1,2\}$; namely, $k_i=\lfloor\frac{\pop(H_i)}{\mathcal{I}}\rceil$, where $\lfloor\cdot\rceil$ denotes rounding to the nearest integer.  It therefore isn't necessary to pass the argument $k$ to the function \GenPlan, since $k$ can be determined from the population of the graph.   In other words, with these settings, \Bonsai is guaranteed to produce a plan with the correct number of districts even without the guardrails that are designed to ensure this.

Although we'll see in the next section that these settings work very effectively at partitioning grid graphs and precinct graphs, the following toy example shows that these settings do not always work.

\begin{example}\label{E:20}
Figure~\ref{F:brainstorm} shows a graph that is a path with $20$ vertices, each with population $9$ or $11$, to be partitioned into $k=20$ districts with tolerance $\epsilon=0.1$.  The only valid plan is the one for which each vertex is a district. The edge labeled $e_2$ has a population of $101$ to its left and $99$ to its right, so the triple $(e=e_2,k_1=10, k_2=10)$ would be chosen by our system of identifying the ``best'' triple as the most balanced one; however, this choice makes it impossible to further partition the left and right pieces.  The triple $(e=e_1,k_1=10, k_2=10)$ would lead to a successful partition, but is less balanced and is only valid with respect to the looser tolerance multiplier function $\phi(n)=n$.   
\end{example}

\begin{figure}[bht!]
\centering
\includegraphics[width=2.3in]{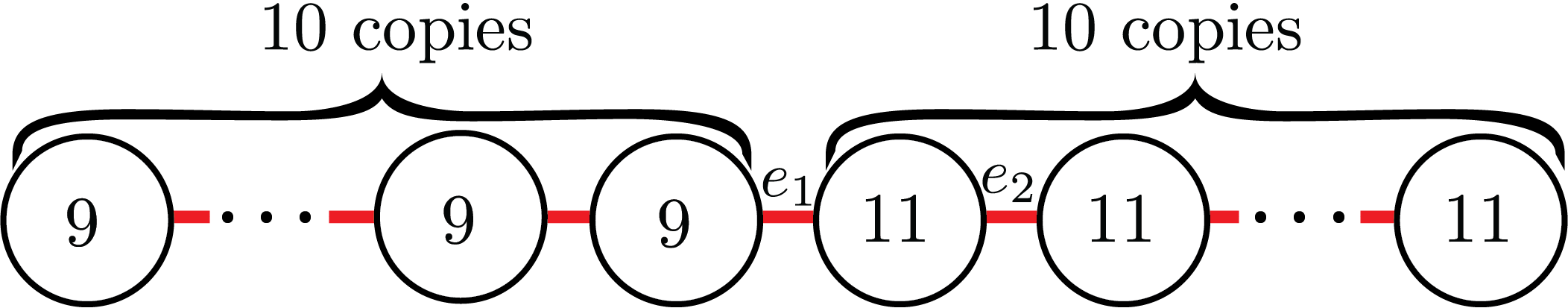}
\caption{A toy example for which Algorithm~\ref{alg:imperfect-balance3} fails when ``best'' is defined as ``most balanced''.}
\label{F:brainstorm}
\end{figure}

We framed Algorithm~\ref{alg:imperfect-balance3} in terms of a general tolerance multiplier function $\phi$ and a general definition of ``best,'' not just to better handle toy examples, but also to achieve the following theoretical guarantee.
\begin{proposition}\label{P:snow}
    If $\phi$ is the identity function ($\phi(n)=n$) and the ``best'' triple is defined by a probabilistic method with full support (that is, each valid triple has a non-zero chance of being chosen), then any valid plan has a non-zero probability of being generated by Algorithm~\ref{alg:imperfect-balance3}. 
\end{proposition}

\begin{proof}
Let $P=\{D_1,...,D_k\}$ be a district plan.  Choose a spanning tree of each district.  Connect them into a single spanning tree $T$ of $\G$ by adding $k-1$ connecting edges, which is equivalent to choosing a spanning tree on the district-level quotient graph $\G/P$.  Arbitrarily choose an order in which these $k-1$ connecting edges will be removed.  Removing them in that order induces a splitting order.  

\Bonsai has a nonzero probability of choosing $T$ as its first spanning tree because it draws the first tree uniformly.  In the subsequent steps, it has a nonzero probability of choosing the connecting edges in the given order.  Thus, it has a nonzero probability of generating $P$ without ever backtracking or drawing any new trees. 
\end{proof}

This proposition represents an advantage of \Bonsai over Markov chain samplers, for which one typically lacks the irreducibility theorems that would be needed in order to guarantee that each valid plan has a nonzero chance of occurring in the sample.

Although Example~\ref{E:20} and Proposition~\ref{P:snow} demonstrate a theoretical advantage of $\phi(n)=n$ over $\phi(n)=1$, experimentation shows that the stricter option $\phi(n)=1$ tends to perform noticeably faster than the looser option $\phi(n)=n$ on actual states because less backtracking is needed.

\subsection{Variations}
As described above, the \Bonsai algorithm can be varied by changing the tolerance multiplier function $\phi$ and/or the algorithm for choosing the ``best'' triple.  

Other variations can be obtained by choosing different methods for drawing random spanning trees.  For instance, uniform spanning trees might be replaced with minimum spanning trees generated by Kruskal's algorithm, as is common in many implementations of ReCom.
Minimum spanning trees are faster to compute and have the advantage that edge weights can be chosen to reflect some kinds of redistricting priorities; for instance, upweighting edges that connect units in different counties has the effect of reducing the number of counties split between districts in the resulting district plans.

\section{Empirical results}

For our empirical study, we compared two variations of \Bonsai (uniform spanning trees and minimum spanning trees, both with tolerance multiplier function $\phi=1$ where applicable) with four variations of ReCom. In addition to the choice of uniform or minimum spanning trees, ReCom offers two options for how to choose a district pair to merge and re-split at each step: either choose a cut edge uniformly at random and choose the districts connected by that edge, or choose a district pair uniformly at random.  This gives rise to four ReCom variants:
\begin{itemize}
    \item ReCom A: minimum spanning trees, cut edge selection;
    \item ReCom B: minimum spanning trees, district pair selection;    
    \item ReCom C: uniform spanning trees, cut edge selection;
    \item ReCom D: uniform spanning trees, district pair selection.    
    \end{itemize}

(We did not consider the reversible variant of ReCom, which is compared to the other ReCom variants in detail in Appendix D of \cite{Cannon_Duchin_Randall_Rule_2022}.)

We chose the following graphs/district partition sizes:
\begin{itemize}
    \item $7 \times 7$ grid graph into $7$ equal-size districts;
    \item $50 \times 50$ grid graph into $10$, $25$ and $50$ equal-size districts;
    \item Pennsylvania 2010 VTD graph (which has 9255 nodes) into 18 Congressional districts (as PA had in the 2010 census cycle) with maximum population deviation  $\epsilon = 0.01$;
    \item North Carolina 2010 VTD graph (which has 2692 nodes) into 120 state House districts with maximum population deviation  $\epsilon = 0.05$.   
\end{itemize}
Our choices of grid graph partitions are the same as those studied in Appendix D of \cite{Cannon_Duchin_Randall_Rule_2022}; our choices of state VTD graphs represent real-world district plans that were the subject of major litigation in recent years.

For each of these scenarios, we generated 6 ensembles of 100,000 plans each using the two variants of \Bonsai and the four variants of ReCom described above.  For the $7\times 7$ grid graph, we constructed an additional ensemble of 100,000 plans using Complete Cut.  We computed the following compactness statistics for each ensemble:
\begin{itemize}
    \item (plan-wide) cut edges;
    \item individual district perimeters (grid graphs only);
    \item  ordered district vote share distributions (VTD graphs only) with respect to the 2016 Presidential election in Pennsylvania and the 2016 Governor's election in North Carolina.
\end{itemize}

\subsection{$7 \times 7$ grid into $7$ districts}

Figure \ref{fig:7x7-into-7-cut-edges-complete-cut} shows histograms comparing the ensemble statistics for plan-wide cut edges for both variants of \Bonsai to those for the analogous versions of \CompleteCut for $7$-district plans on a $7 \times 7$ grid.  In both cases, we see that cut edge statistics for \Bonsai and \CompleteCut are similar, with \Bonsai having a slightly higher ensemble average (corresponding to slightly less compact districts) than \CompleteCut.  As we might expect, the minimum spanning tree versions produce slightly more compact districts than the uniform spanning tree versions.

\begin{figure}[bht!]
\begin{center}
\caption{Ensemble statistics for plan-wide cut edges for \protect\Bonsai vs. \protect\CompleteCut, (a) using minimum spanning trees, (b) using uniform spanning trees, for $7$-district plans on a $7 \times 7$ grid}
\label{fig:7x7-into-7-cut-edges-complete-cut}
\subfloat[]{\includegraphics[width=2in]{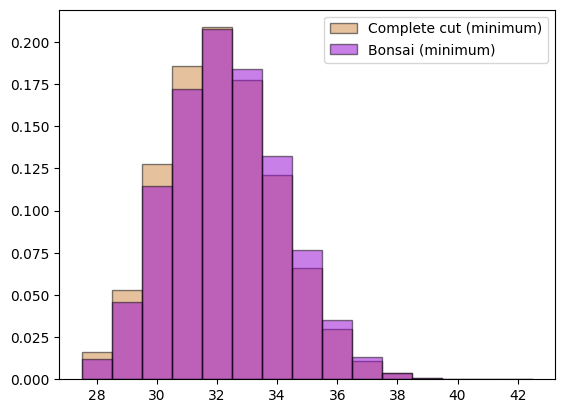} }
\ 
\subfloat[]{\includegraphics[width=2in]{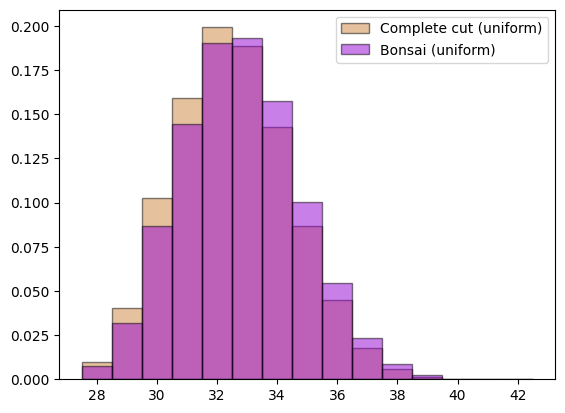}  }
\end{center}
\end{figure}

Figure \ref{fig:7x7-into-7-cut-edges-recom} shows similar histograms comparing the ensemble statistics for plan-wide cut edges for the minimum spanning tree variant of \Bonsai to the minimum spanning tree variants of ReCom (ReCom A and ReCom B), and for the uniform spanning tree variant of \Bonsai to the uniform spanning tree variants of ReCom (ReCom C and ReCom D), for $7$-district plans on a $7 \times 7$ grid.
For both types of spanning trees, we see that cut edge statistics for \Bonsai lie somewhere between those for the two variants of ReCom, and closer to the statistics for the district pair selection variant (ReCom B and ReCom D, respectively).

\begin{figure}[bht!]
\begin{center}
\caption{Ensemble statistics for plan-wide cut edges for \protect\Bonsai vs. ReCom variants, (a) using minimum spanning trees, (b) using uniform spanning trees, for $7$-district plans on a $7 \times 7$ grid}\label{F:7_cut_edges}
\label{fig:7x7-into-7-cut-edges-recom}
\subfloat[]{\includegraphics[width=2in]{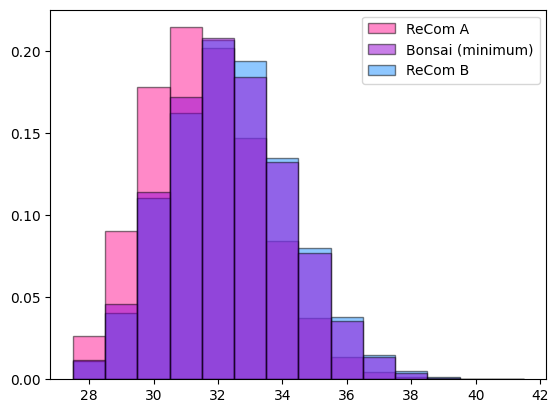}  }
\
\subfloat[]{\includegraphics[width=2in]{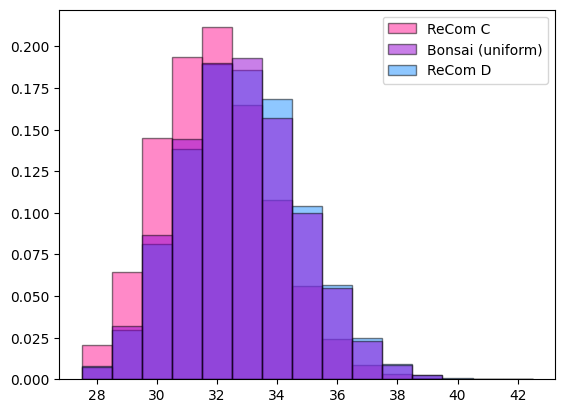}  }
\end{center}
\end{figure}

For $7$-district plans on a $7 \times 7$ grid, individual district perimeters can only take on three possible values: $12$, $14$, and $16$.  Figure 
\ref{fig:7x7-into-7-perims} shows histograms comparing the ensemble statistics for (a) all four minimum spanning tree ensembles, and (b) all four uniform spanning tree ensembles, for the perimeters of all districts in each ensemble; that is, each histogram illustrates the distribution of the set of perimeters of all of the districts in all of the maps of the ensemble.  These district-level perimeter statistics agree with the plan-wide cut edges statistics in the sense that for both types of spanning trees, district perimeter statistics for \Bonsai lie somewhere between those for the two variants of ReCom, and closer to the statistics for the district pair selection variant (ReCom B and ReCom D, respectively).

\begin{figure}[bht!]
\begin{center}
\caption{Ensemble statistics for district perimeters for \protect\Bonsai vs. Complete Cut and ReCom variants, (a) using minimum spanning trees, (b) using uniform spanning trees, for $7$-district plans on a $7 \times 7$ grid}\label{F:7_ditrict_perimeters}
\label{fig:7x7-into-7-perims}
\subfloat[]
{\includegraphics[width=2in]{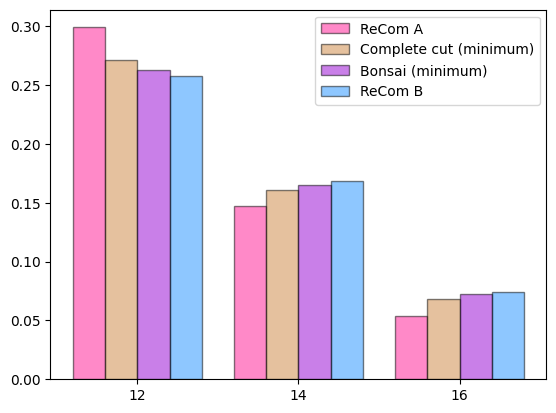} }
\ 
\subfloat[]
{\includegraphics[width=2in]{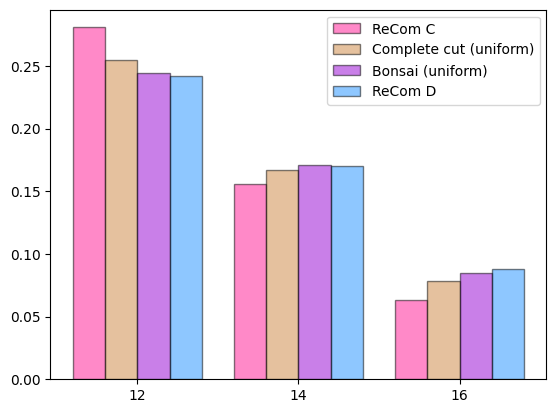}  }
\end{center}
\end{figure}

\FloatBarrier

The appendix contains similar plots showing that the $50$-by-$50$ grid behaves similarly to the $7$-by-$7$ grid, with \Bonsai statistics falling between the two corresponding variants of ReCom.
\FloatBarrier

\subsection{Pennsylvania and North Carolina VTD ensembles}

We next created ensembles of $18$-district plans in Pennsylvania and $120$-district plans in North Carolina, using 2010 VTD shape files and allowing a maximum population deviation of $\epsilon = 0.01$ for Pennsylvania and $\epsilon=0.05$ for North Carolina.  

Figure \ref{F:PA_NC} compares histograms of plan-wide cut edges, just as in Figures~\ref{F:7_cut_edges} and~\ref{F:cut_edges_combined}.  As in the previous cases, we see that for both types of spanning trees, cut edge statistics for \Bonsai lie somewhere between those for the two variants of ReCom, and closer to the statistics for the district pair selection variant (ReCom B and ReCom D, respectively).

\begin{figure}[hbt!]
    \centering
    \caption{Ensemble statistics for plan-wide cut edges for \protect\Bonsai vs. ReCom variants for PA ($k=18$) and NC ($k=120$)}
\begin{tikzpicture}

\node at (0,10.1) {{\footnotesize Minimum spanning trees}};

\node at (5.5,10.1) {{\footnotesize Uniform spanning trees}};

\node at (-3.2,8.2) {{\scriptsize PA}};

\node at (-3.2,7.9){{\scriptsize $(k=18)$}};

\node at (9,8) {$\ \ \ \ \ \ $};

\node at (0,8) {\includegraphics[width=2in]{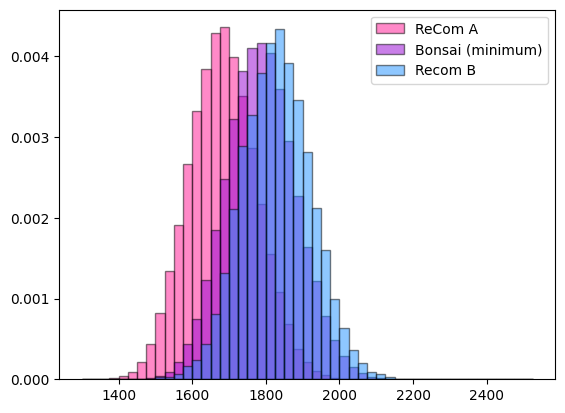}};

\node at (5.5,8) {\includegraphics[width=2in]{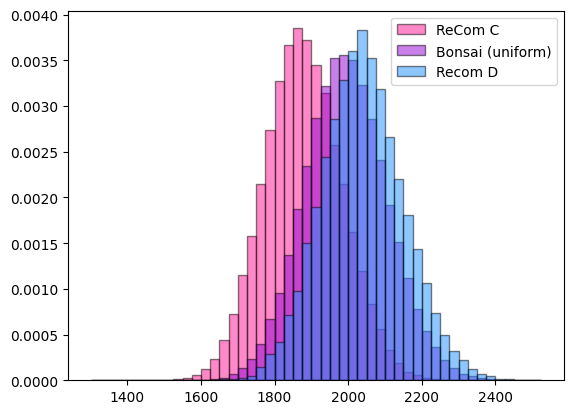}};

\node at (-3.2,4.2) {{\scriptsize NC}};

\node at (-3.2,3.9) {{\scriptsize $(k=120)$}};

\node at (9,4) {$\ \ \ \ \ \ $};

\node at (0,4) {\includegraphics[width=2in]{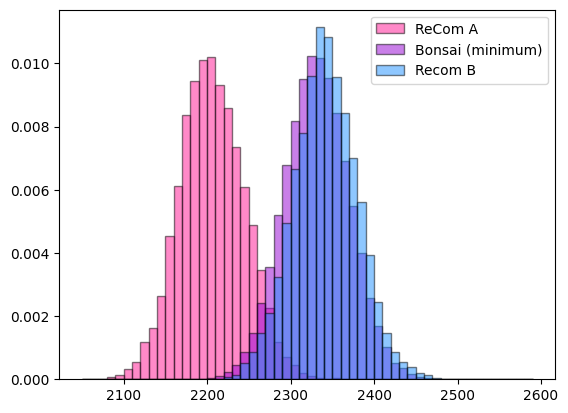}};

\node at (5.5,4) {\includegraphics[width=2in]{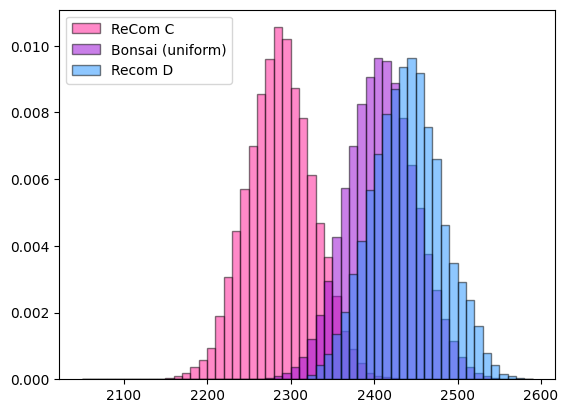}};

\end{tikzpicture}
    \label{F:PA_NC}
\end{figure}

We next consider partisan data.  Figure~\ref{fig:PA-into-18-PRES16} shows boxplots for Democratic vote share by district with respect to the 2016 Presidential election in Pennsylvania.  Figure~\ref{fig:NC-into-120-GOV16-b} in the appendix is a similar plot for the 2016 Governor's election in North Carolina.  For each plot, the districts are first ordered by their Democratic vote share, and the horizontal axis represents these ordered districts.  For example, the boxplots above position 1 on the $x$-axis represent the range of Democratic vote shares in the reddest (most Republican-leaning) districts in all of the plans of the ensemble.  

The statistics for both versions of \Bonsai are extremely close to those for all four versions of ReCom.

\begin{figure}[bht!]
\begin{center}
\caption{Ensemble statistics for Democratic vote share by district in the 2016 Presidential election for \protect\Bonsai vs. ReCom variants, (a) using minimum spanning trees, (b) using uniform spanning trees, for $18$-district plans on 2010 Pennsylvania VTDs}
\label{fig:PA-into-18-PRES16}
\subfloat[]
{\includegraphics[width=4.5in]{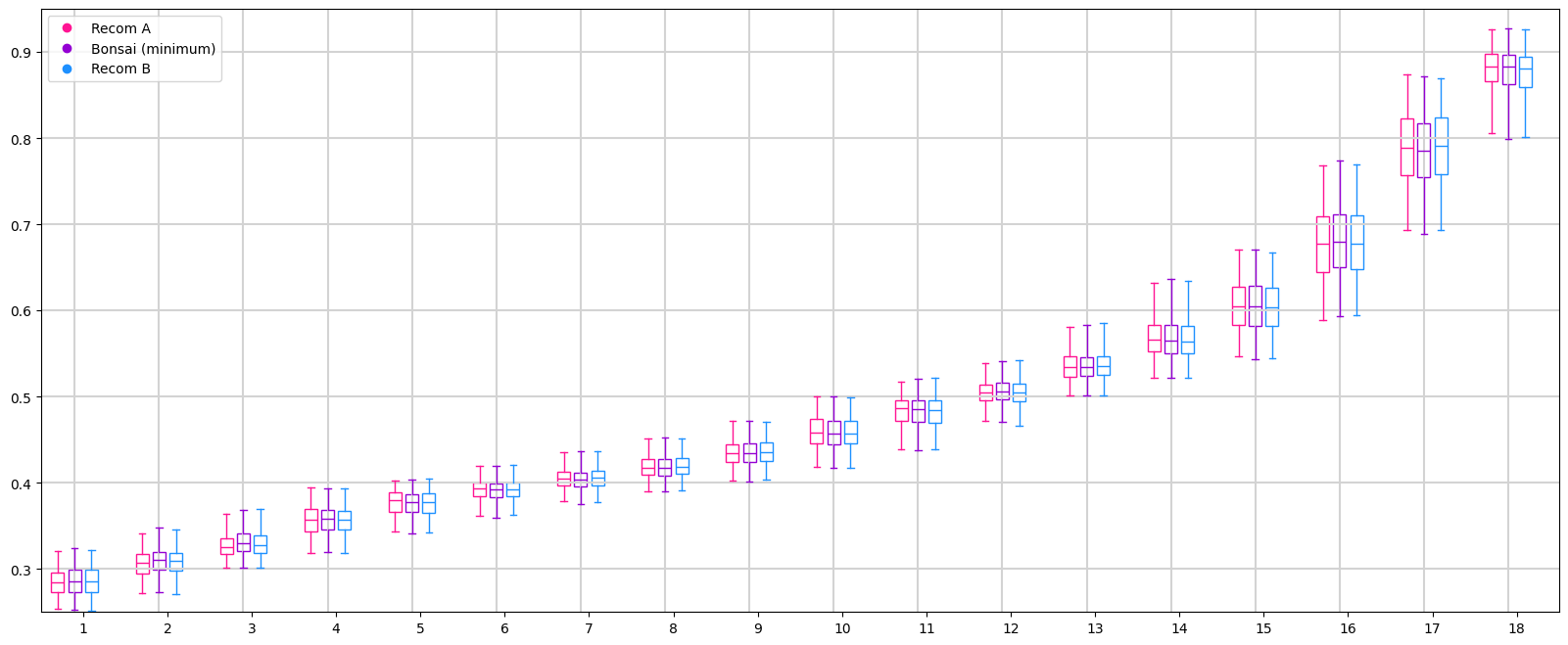} }
\ 
\subfloat[]
{\includegraphics[width=4.5in]{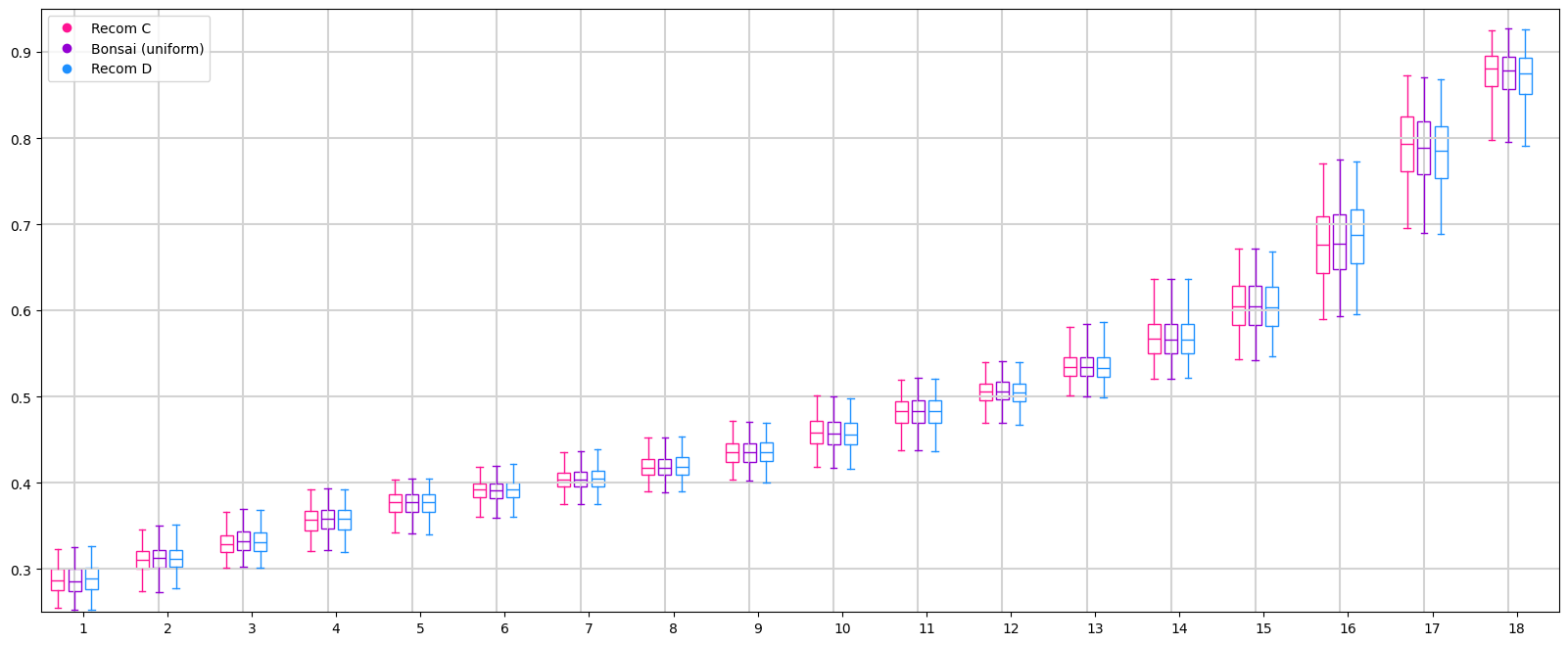}  }
\end{center}
\end{figure}

\FloatBarrier

\section{Conclusion}

We have introduced a class of algorithms for constructing ensembles of district plans via \emph{independent} sampling of connected graph partitions. In contrast to widely used Markov chain approaches such as ReCom, our \Bonsai algorithm produces plans one at a time without relying on long-run mixing, thereby avoiding concerns about ergodicity, slow mixing, and autocorrelation. Independent sampling allows for full parallelization and ensures that the effective sample size equals the actual sample size -- an important practical and statistical advantage.

From a theoretical standpoint, we explicitly described the sampling distribution of the perfect-balance version of Bonsai without backtracking. For imperfect balance, we introduced a flexible framework based on a tolerance multiplier function and a principled rule for selecting among the valid cuts. 

Our empirical results on grid graphs and on VTD graphs for Pennsylvania and North Carolina demonstrate several consistent patterns. Across all scenarios tested, ensemble statistics produced by \Bonsai lie between those generated by the two principal ReCom variants (cut-edge selection and district-pair selection), and typically closer to the district-pair selection variant. This behavior is stable across choices of uniform versus minimum spanning trees, across increasing numbers of districts, and across metrics. In particular, \Bonsai and ReCom yield very similar district-level vote share distributions, suggesting that \Bonsai produces substantively comparable baselines while offering the computational and statistical advantages of independence.

The results also highlight a conceptual point: many ensemble-level statistics appear robust to substantial differences in sampling methodology, at least within the family of spanning-tree-based approaches. This robustness strengthens confidence in ensemble analysis as a tool for evaluating enacted maps, while also underscoring the value of having multiple algorithmic paradigms available for cross-validation.

Overall, \Bonsai provides a practical and mathematically transparent framework for independent sampling of graph partitions, offering an effective alternative to Markov chain–based redistricting algorithms and expanding the toolkit available for ensemble analysis in research and litigation contexts.

\bibliographystyle{alpha}  
\bibliography{bibliography} 

\appendix
\section{Importance Sampling Empirics}\label{sec:Importance_eprirics}
In this section, we provide evidence that importance sampling via \SimCut, as described in Section~\ref{sec:Importance}, works effectively in grid graphs.

We first created ensembles of size $N=10,000$ of the $7$-by-$7$ grid partitioned into $k=7$ districts via \CompleteCut and \SimCut.  We additionally used importance weights with \SimCut to target the same distribution from which \CompleteCut draws, which is described in Proposition~\ref{P:completecut}.  The required weights are obtained from Equation~\ref{eq:simplified-importance-weight} by removing the factor $ST(\G/\mP)$ from the denominator.

Figure~\ref{F:7_by_7} indicates that, while the unweighted \SimCut algorithm draws from a distribution that appears fairly similar to that of \CompleteCut, at least with respect to the cut edge statistic, the importance weights shift this distribution to become almost identical to that of \CompleteCut; in fact, the total variation distance is about $0.017$.

\begin{figure}[bht!]
\centering\includegraphics[width=3.5in]{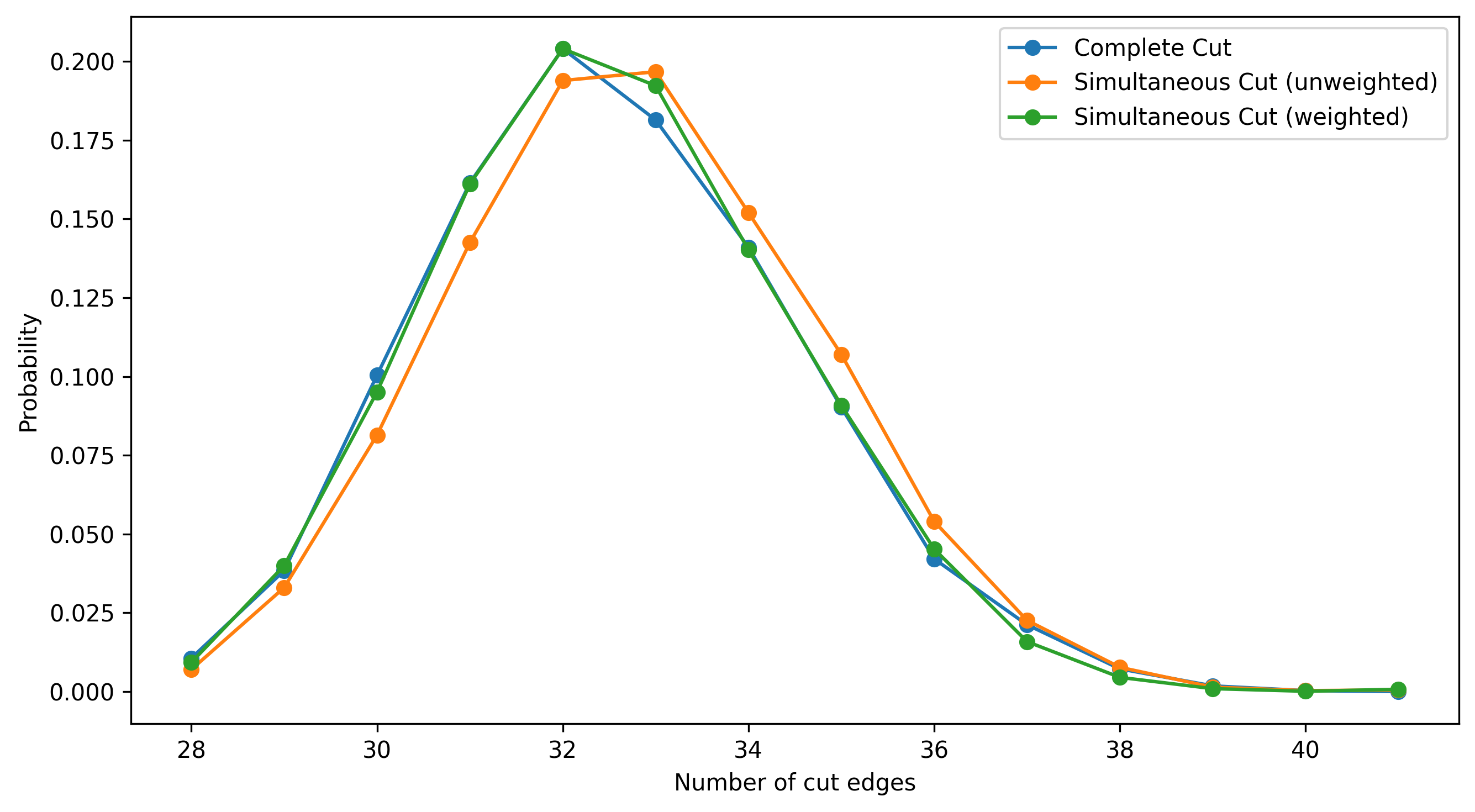} 
\caption{The cut edge distribution for ensembles of size 10000 of the $7$-by-$7$ grid partitioned into $7$ districts via \protect\CompleteCut, \protect\SimCut, and \protect\SimCut with importance weights targeting the same distribution as \protect\CompleteCut.} \label{F:7_by_7}
\end{figure}

Next, for several choices of $n,k$, we considered the $n$-by-$n$ grid graph partitioned into $k$ districts.  We used \SimCut to build ensembles of $N=500$ pairs of the form $(\mP, \SO)$.  For each ensemble, we recorded the weights $\{w_1,...,w_N\}$ and computed the effective sample size
$$ N_{\text{eff}} = \frac{\left( \sum_{i=1}^N w_i \right)^2}{\sum_{i=1}^N w_i^2}.$$
Note that $N_{\text{eff}}/N\in[0,1]$.  Values closer to $1$ indicate effective sampling with mild weight variation, whereas values closer to $0$ indicate ineffective sampling with large variation and hence a larger amount of weight concentrated in a smaller number of plans.  

Table~\ref{tab:importance-sampling} reports the results.  The data indicates that importance sampling is extremely effective for $k=2$, with effective sample sizes about $94$–$96\%$ of the full sample size regardless of grid size. The efficiency tends to decline as the number of districts increases, with substantially more variability in the weights; nevertheless, even the most difficult case tested, $20\times 20$ with $k=10$ districts, has an effective sample size of about $12.5\%$ of the full sample size.

\begin{table}[t]
\centering
\begin{tabular}{rrcrrr}
\toprule
$n$ & $k$ & District area
& $N_\text{eff}/N$
& $w_{95}/w_{5}$
& $w_{\max}/w_{\min}$ \\
\midrule
6  & 2  & 18  & 0.944 & 2.17    & 3.00 \\
6  & 6  & 6   & 0.258 & 34.70   & $1.66\times 10^3$ \\
\midrule
10 & 2  & 50  & 0.964 & 1.92    & 3.70 \\
10 & 5  & 20  & 0.555 & 16.72   & 450.26 \\
10 & 10 & 10  & 0.318 & 68.53   & $2.77\times 10^3$ \\
\midrule
20 & 2  & 200 & 0.965 & 1.94    & 3.52 \\
20 & 5  & 80  & 0.195 & 61.73   & $1.92\times 10^3$ \\
20 & 10 & 40  & 0.125 & $2.75\times 10^3$
                         & $1.66\times 10^6$ \\
\bottomrule
\end{tabular}
\caption{Performance of importance sampling on $n\times n$ grid graphs
with $k$ districts. Each ensemble contains $N=500$ plans.
The columns $w_{95}/w_5$ and $w_{\max}/w_{\min}$ report,
respectively, the ratio of the $95$th to $5$th percentile importance
weights and the ratio of the largest to smallest observed importance
weights.}
\label{tab:importance-sampling}
\end{table}

\section{Computation time and sample quality}

Because Bonsai constructs a complete plan at each step while ReCom changes only two districts at a time, we should expect Bonsai to require more time than ReCom to construct similarly-sized ensembles.  On the other hand, because Bonsai samples independently while ReCom ensembles naturally have a high degree of autocorrelation, we might hope that Bonsai would produce high-quality samples with many fewer plans than a ReCom chain requires.

In Table \ref{table:computation_times}, we compare the time required to construct a 1,000-plan ensemble for each of our examples using Bonsai and ReCom.  These computations were performed on a 2023 MacBook Pro in Python 3.12.  Because our current implementation of Bonsai is Python-only, we used the Python version of GerryChain 0.3.2 (which uses ReCom A by default) and the minimum spanning tree version of Bonsai for these computations.

\begin{table}[bht!]
    \centering
    \caption{Computation times (in seconds) to construct 1,000-plan ensembles using Python versions of minimum spanning tree Bonsai and ReCom A}
\begin{tabular}{|c|c|c|c|c|}
\hline
Graph  & Number of  & Bonsai  & ReCom  & Bonsai:ReCom \\[-0.02in]
 & districts & time & time & ratio\\
\hline
$7 \times 7$ & 7 & 0.43 & 0.43 & 1.0 \\
\hline
$50 \times 50$ & $10$ & 250.68 & 42.74 & 5.87 \\
\hline
$50 \times 50$ & $25$ & 217.86 & 11.86 & 18.37 \\
\hline
$50 \times 50$ & $50$ & 177.47 & 5.01 & 35.42 \\
\hline PA VTDs & 18 & 2,298.51  & 89.86 & 25.58\\
\hline NC VTDs & 120 & 444.71 &  2.41 & 184.53\\
\hline
    \end{tabular}
    \label{table:computation_times}
\end{table}

Unsurprisingly, the ratio of Bonsai time to ReCom time is highest for plans with large numbers of districts, where ReCom makes the smallest changes at each step and consequently takes much less time per step.  Interestingly, the three ensembles on the $50 \times 50$ grid show that Bonsai is also faster for plans with more districts on the same graph, but not as dramatically so as ReCom.

The final row of Table~\ref{table:computation_times} shows that Bonsai takes about $185$ times as long to construct one plan as ReCom takes to advance one chain step, which is a surprisingly low ratio.  It is shown in~\cite{Tapp_Proebsting_Ramsay_2025} that ReCom takes more than $2500$ steps to mix for NC statehouse plans\footnote{More precisely, Table 4 of~\cite{Tapp_Proebsting_Ramsay_2025} shows that a ReCom-A ensemble of NC statehouse plans that sub-samples every $2500^\text{th}$ plan still has an autocorrelation of $0.08$ with respect to certain metrics, and hence an approximate sample size of only about $85\%$ of the actual sample size. This indicates that subsampling to approximate independent sampling would require steps of size somewhat greater than $2500$.}, so any ratio less than $2500$ should be considered a win for Bonsai over ReCom. 

For a measure of sample quality, we compared the number of unique districts occurring in 100,000-plan ensembles constructed by the minimum spanning tree version of Bonsai and ReCom A for each of our examples.  These results are shown in Table \ref{table:unique_districts}.

\begin{table}[bht!]
    \centering
    \caption{Numbers of unique districts in 100,000-plan ensembles constructed by minimum spanning tree Bonsai and ReCom A}
\begin{tabular}{|c|c|c|c|c|}
\hline
Graph  & Number of  & Bonsai unique & ReCom unique  & Bonsai:ReCom \\[-0.02in]
 & districts & districts & districts & ratio\\
\hline
$7 \times 7$ & 7 & 10,829 & 7,603 & 1.42 \\
\hline
$50 \times 50$ & $10$ & 1,000,000 & 196,444 & 5.09 \\
\hline
$50 \times 50$ & $25$ & 2,499,976 & 187,344 & 13.34 \\
\hline
$50 \times 50$ & $50$ & 4,982,596 & 172,059 & 28.96 \\
\hline PA VTDs & 18 & 1,800,000  & 199,131 & 9.04 \\
\hline NC VTDs & 120 & 10,508,410 &  148,450 & 70.79 \\
\hline
    \end{tabular}
    \label{table:unique_districts}
\end{table}

Comparing the final columns of Tables \ref{table:computation_times} and \ref{table:unique_districts} gives some sense of how well the extra computation time for Bonsai vs. ReCom is balanced by the greater variety of districts produced by Bonsai.  For the grid graph ensembles, these ratios are fairly comparable.  The disparity between the computation time and district variety ratios is greater for the PA and NC examples, but still within a factor of 3---which could be compensated for by a modest parallelization of Bonsai, if desired. 
This would seem to indicate that the computational practicality of Bonsai is comparable to that of ReCom.

\section{The $50 \times 50$ grid}

In this appendix, we illustrate ensemble statistics for plan-wide cut edges (Figure~\ref{F:cut_edges_combined}) and district perimeters (Figure~\ref{F:district_perimeters_combined}) for $k$-district plans on the $50$-by-$50$ grid with $k\in\{10, 25, 50\}$.  These plots match Figures~\ref{F:7_cut_edges} and~\ref{F:7_ditrict_perimeters} respectively, except that \CompleteCut is omitted from the plots because its rejection rate is untenably high for graphs of this size.  In particular, as in the previous figures, the plots on the left compare the minimum spanning tree variant of \Bonsai to the minimum spanning tree variants of ReCom (ReCom A and ReCom B), while the plots on the right compare the uniform spanning tree variant of \Bonsai to the uniform spanning tree variants of ReCom (ReCom C and ReCom D).  

As in the $7$-by-$7$ case, we see that for both types of spanning trees and for both metrics, \Bonsai lies somewhere between the two variants of ReCom, and closer to the district pair selection variant (ReCom B and ReCom D, respectively).

\begin{figure}[hbt!]
    \centering
    \caption{Ensemble statistics for plan-wide cut edges for \protect\Bonsai vs. ReCom variants for $k$-district plans on a $50 \times 50$ grid with $k\in\{10, 25, 50\}$}
\begin{tikzpicture}

\node at (0,10.1) {{\footnotesize Minimum spanning trees}};

\node at (5.5,10.1) {{\footnotesize Uniform spanning trees}};

\node at (-3.2,8) {{\scriptsize $k=10$}};

\node at (9,8) {$\ \ \ \ \ \ $};

\node at (0,8) {\includegraphics[width=2in]{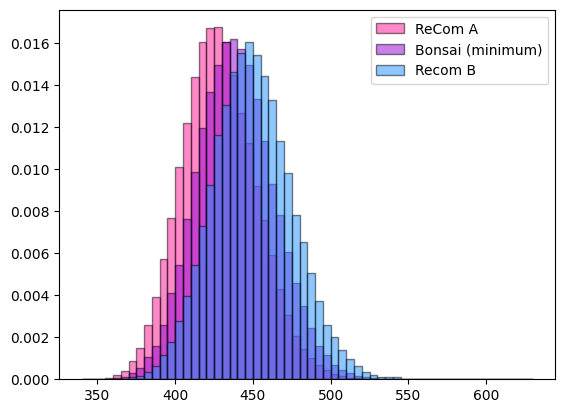}};

\node at (5.5,8) {\includegraphics[width=2in]{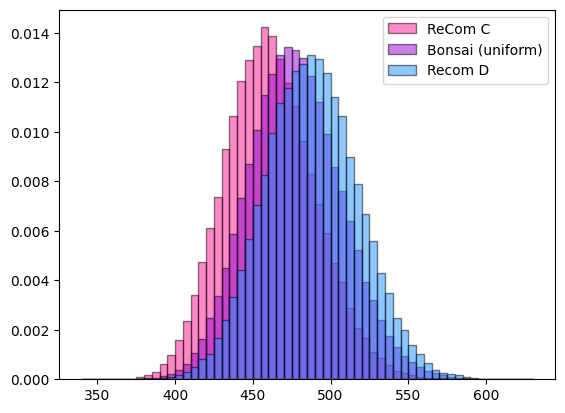}};

\node at (-3.2,4) {{\scriptsize $k=25$}};

\node at (9,4) {$\ \ \ \ \ \ $};

\node at (0,4) {\includegraphics[width=2in]{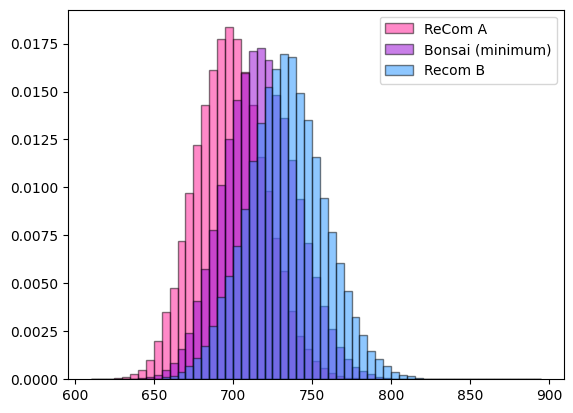}};

\node at (5.5,4) {\includegraphics[width=2in]{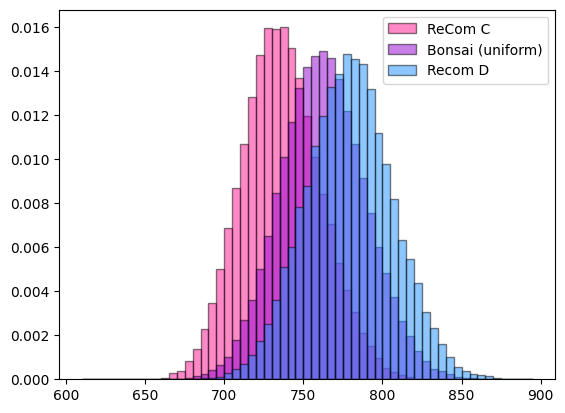}};

\node at (-3.2,0) {{\scriptsize $k=50$}};

\node at (9,0) {$\ \ \ \ \ \ $};

\node at (0,0) {\includegraphics[width=2in]{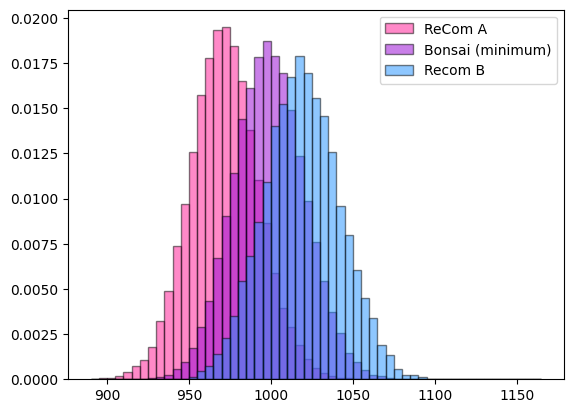}};

\node at (5.5,0) {\includegraphics[width=2in]{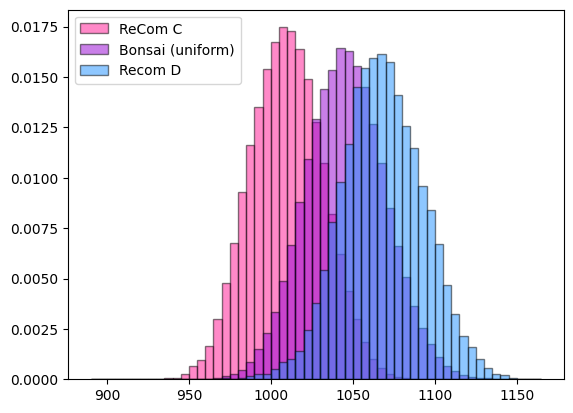}};
\end{tikzpicture}
    \label{F:cut_edges_combined}
\end{figure}

\begin{figure}[hbt!]
    \centering
    \caption{Ensemble statistics for district perimeters for \protect\Bonsai vs. ReCom variants for $k$-district plans on a $50 \times 50$ grid, with $k\in\{10, 25, 50\}$}
\begin{tikzpicture}

\node at (0,10.1) {{\footnotesize Minimum spanning trees}};

\node at (5.5,10.1) {{\footnotesize Uniform spanning trees}};

\node at (-3.2,8) {{\scriptsize $k=10$}};

\node at (9,8) {$\ \ \ \ \ \ $};

\node at (0,8) {\includegraphics[width=2in]{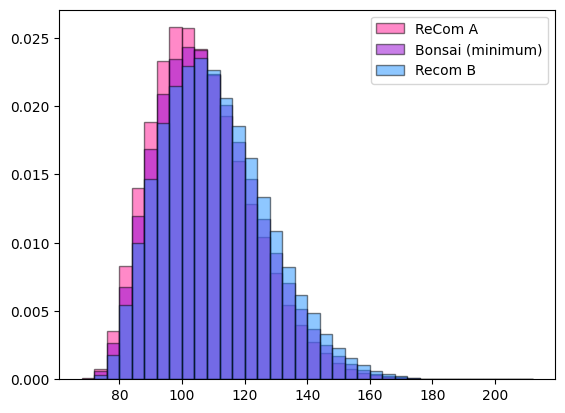}};

\node at (5.5,8) {\includegraphics[width=2in]{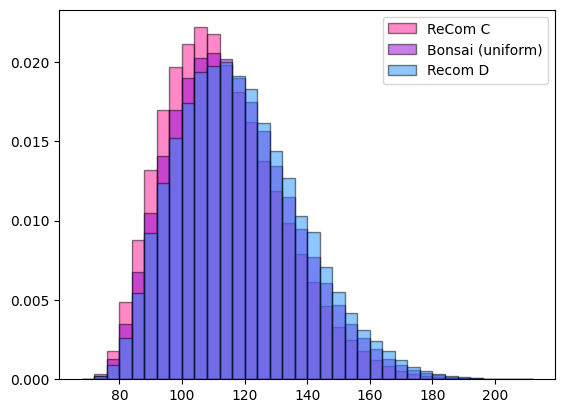}};

\node at (-3.2,4) {{\scriptsize $k=25$}};

\node at (9,4) {$\ \ \ \ \ \ $};

\node at (0,4) {\includegraphics[width=2in]{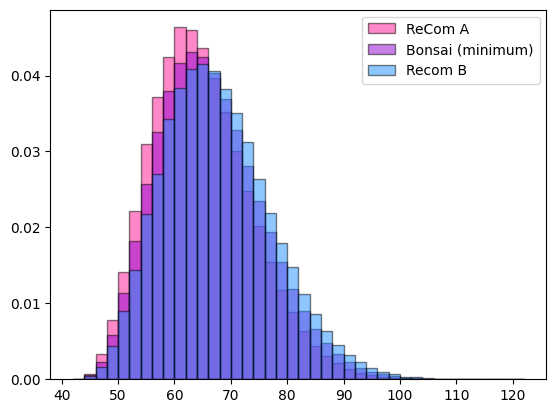}};

\node at (5.5,4) {\includegraphics[width=2in]{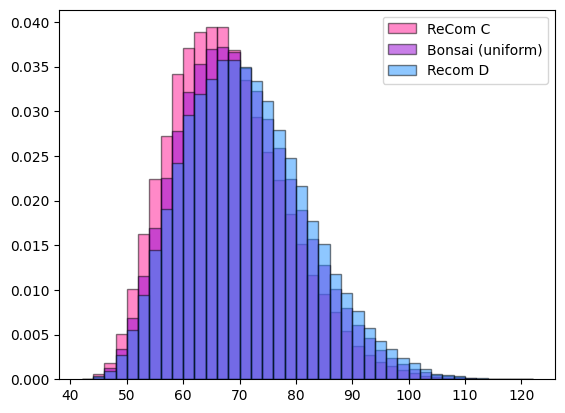}};

\node at (-3.2,0) {{\scriptsize $k=50$}};

\node at (9,0) {$\ \ \ \ \ \ $};

\node at (0,0) {\includegraphics[width=2in]{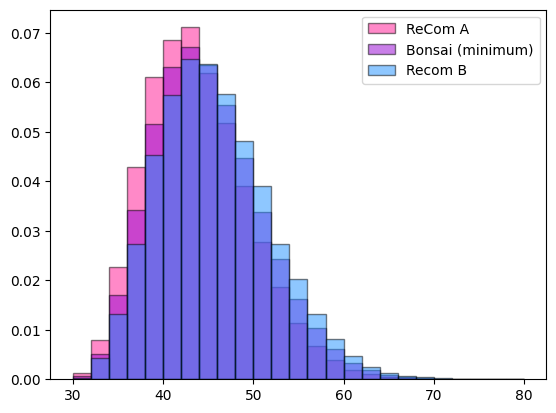}};

\node at (5.5,0) {\includegraphics[width=2in]{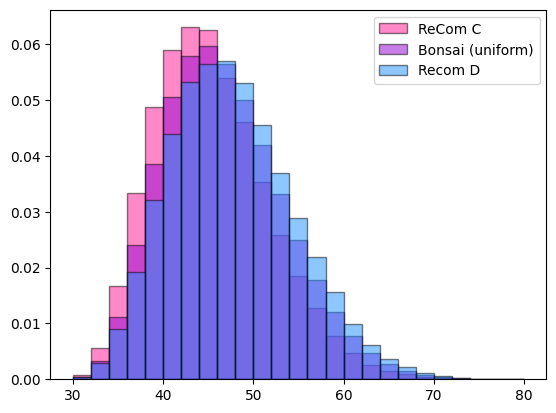}};
\end{tikzpicture}
    \label{F:district_perimeters_combined}
\end{figure}

\section{North Carolina Plots}
Similar to Figure~\ref{fig:PA-into-18-PRES16}, Figure~\ref{fig:NC-into-120-GOV16-b} shows the Democratic vote share by ordered district for $120$-district plans of North Carolina, with respect to the 2016 Governor's election, but cropped to show only the middle third of the ordered districts (numbers 41-80).  In all cases (including the omitted top and bottom thirds of the North Carolina plots), the statistics for both versions of \Bonsai are extremely close to those for all four versions of ReCom.

\begin{figure}[bht!]
\begin{center}
\caption{Ensemble statistics for Democratic vote share by district in the 2016 Governor's election for \protect\Bonsai vs. ReCom variants, (a) using minimum spanning trees, (b) using uniform spanning trees, for $120$-district plans on 2010 North Carolina VTDs; these plots show the middle 40 districts.}
\label{fig:NC-into-120-GOV16-b}
\subfloat[]
{\includegraphics[width=4.5in]{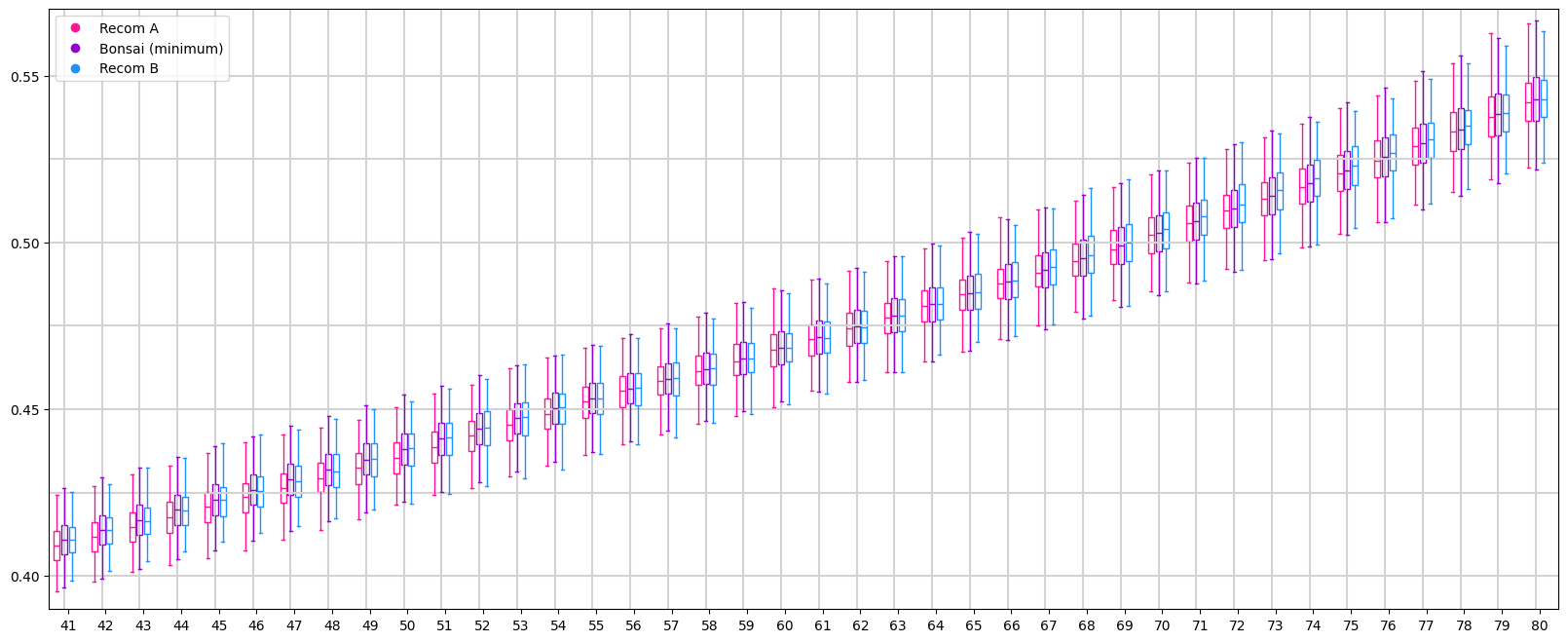} }
\ 
\subfloat[]
{\includegraphics[width=4.5in]{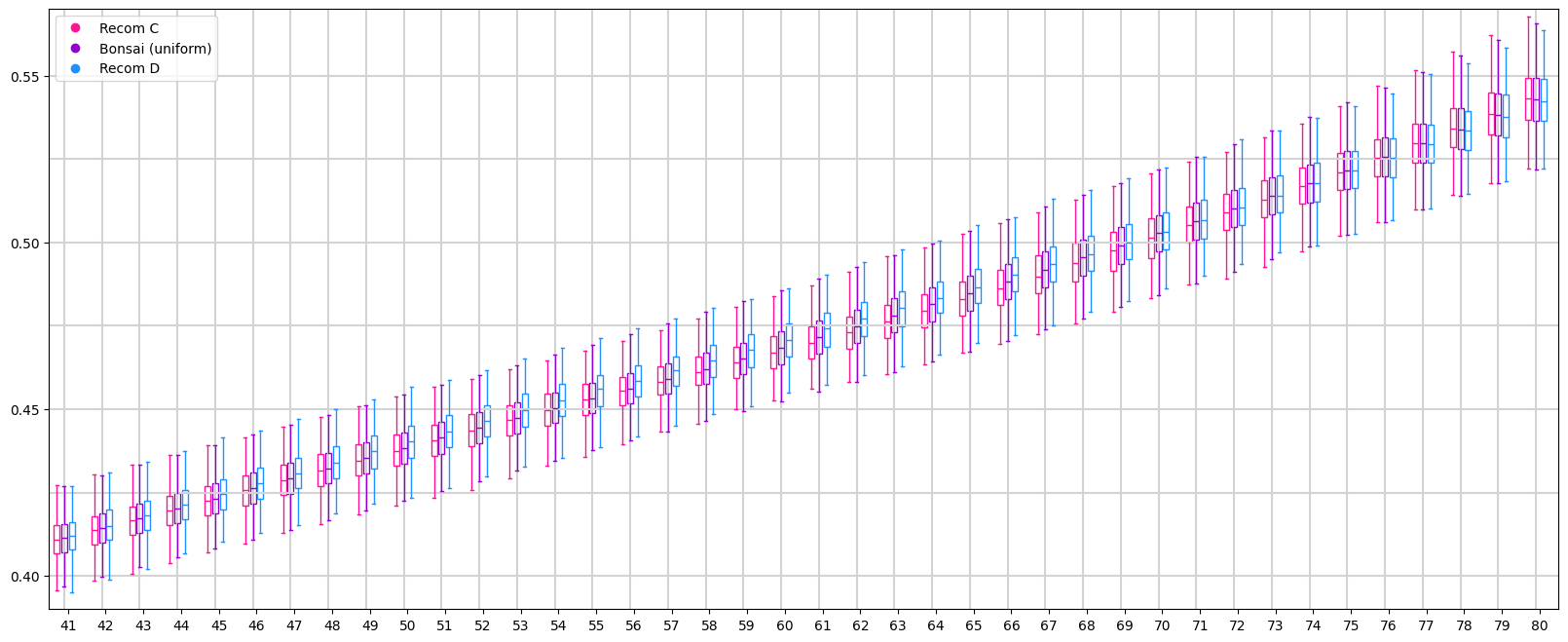}  }
\end{center}
\end{figure}

\end{document}